\documentclass[12pt]{article}
\usepackage{color,latexsym,amsfonts,amssymb}
\usepackage{hyperref}
\usepackage{amsmath,cite}
\usepackage{enumitem}

\allowdisplaybreaks[4]
\usepackage{amsthm}
\usepackage[a4paper, margin=1in]{geometry}
\usepackage{float}
\usepackage{siunitx}  

\newcommand{\E}{\mathbb{E}} 

\newcommand{\R}{\mathbb{R}}

\def\pr{\mathbb{P}}

\newcommand{\mean}{{\E}}

\def\var{{\rm Var}}

\def\pR{{\R_+^\circ}}

\def\ignore#1{}

\def\1{{\bf 1}}

\newcommand{\beas}{\begin{eqnarray*}}
\newcommand{\enas}{\end{eqnarray*}}

\newcommand{\bea}{\begin{eqnarray}}
\newcommand{\ena}{\end{eqnarray}}

\def\eq{\begin{equation}}
\def\en{\end{equation}}

\numberwithin{equation}{section}
\newtheorem{Theorem}{Theorem}[section]

\newtheorem{Lemma}[Theorem]{Lemma}

\newtheorem{Setup}[Theorem]{Setup}

\theoremstyle{definition}

\newtheorem{Example}[Theorem]{Example}

\makeatletter
\renewcommand\theequation{\thesection.\@arabic\c@equation}
\makeatother

\def\I{\mathbb{I}}
\def\IR{\R}
\def\IN{\mathbb{N}}

\def\be#1{\begin{equation*}#1\end{equation*}}
\def\ben#1{\begin{equation}#1\end{equation}}
\def\eq#1{\eqref{#1}}
\def\ba#1{\begin{align*}#1\end{align*}}
\def\ban#1{\begin{align}#1\end{align}}
\newcommand\cov{\textrm{Cov}}

\def\eps{\varepsilon}

\def\mcI{\mathcal{I}}
\def\dtv{d_{\textrm{TV}}}
\def\law{\mathcal{L}}
\def\mcF{\mathcal{F}}
\def\t#1{^{(#1)}}

\newcommand{\mvert}{\, \vert \,}
\newcommand{\tb}{\widetilde{b}}
\newcommand{\tx}{\tilde{x}}

\newcommand{\tD}{\widetilde{D}}

\newcommand{\tZ}{\widetilde{Z}}

\newcommand{\NN}{\mathbb{N}}
\newcommand{\Rplus}{\R_{+}}

\def\norm#1{\|#1\|}

\newcommand{\pop}{\mathrm{Pop}}
\newcommand{\done}{\bar{d}_1}
\newcommand{\dtwo}{\bar{d}_2}

\newcommand{\mcb}{\mathcal{B}}
\newcommand{\mcf}{\mathcal{F}}

\newcommand{\mcn}{\mathcal{N}}
\newcommand{\mcx}{\mathcal{X}}

\newcommand{\mfn}{\mathfrak{N}}

\newcommand{\ftv}{\mcf_{\textrm{TV}}}
\newcommand{\fw}{\mcf_{\textrm{W}}}

\newcommand{\Psit}{\Psi \vert_t}
\newcommand{\Psiti}{\Psi^{(i)} \vert_t}
\newcommand{\Upsilont}{\Upsilon \vert_t}

\newcommand{\Rd}{\R^d}
\newcommand{\hbit}{\hspace*{1.5pt}}
\newcommand{\tvnorm}[1]{\| #1 \|_{\textrm{TV}}}
\newcommand{\im}{\mathrm{im}}
\newcommand{\inlawto}{\stackrel{\mathcal{D}}{\raisebox{0pt}[3.5pt][0pt]{$\longrightarrow$}}}

\newcommand{\diam}{\mathrm{diam}}
\DeclareMathOperator{\sd}{sd} 

\usepackage{titlesec}

\titleformat*{\section}{\large\bfseries}
\titleformat*{\subsection}{\itshape}
\titleformat*{\subsubsection}{\itshape}
\usepackage{caption}
\usepackage{booktabs}
\usepackage{graphicx}
\usepackage{subcaption}
\usepackage{float}
\begin{document}

\title{
Wireless network signals with moderately correlated shadowing
still appear Poisson
}

\author{Nathan Ross\thanks{School of Mathematics and Statistics, the University of Melbourne, 
Parkville, VIC 3010, Australia; nathan.ross@unimelb.edu.au} \,
 and
Dominic Schuhmacher\thanks{Institute for Mathematical Stochastics, University of G\"ottingen,
Goldschmidtstra{\ss}e 7, 37077 G\"ottingen, Germany; dschuhm1@uni-goettingen.de}
}

\date{\today}

\maketitle

\begin{abstract}
We consider the point process of signal strengths emitted from transmitters in a 
wireless network and observed at a fixed position. In our model, transmitters are placed 
deterministically or randomly according to a hard core or Poisson point process and signals
are subjected to power law propagation loss and random propagation effects that may be
correlated between transmitters.

We provide bounds on the distance
between the point process of signal strengths and a Poisson process
with the same mean measure, assuming correlated log-normal shadowing.
For ``strong shadowing" and moderate correlations, we find that the 
signal strengths are close to a Poisson process, generalizing
a recently shown analogous result for independent shadowing.
\end{abstract}

\section{Introduction}

In a wireless network, transmitters are placed in some configuration
and emit signals to users of the system (e.g., Wi-Fi or mobile phone).
Understanding the spectrum of signal strengths received at a fixed location in such networks is 
crucial for analysis and design. A main approach to this problem
is to study the behavior of the signal spectrum in
realistic mathematical models of such networks.
(We use the term signal spectrum to mean the point process of 
signal strengths; the exact nature of the signals, e.g., interference
or propagation, is not important to our results.) 

Due to the increasing prevalence of wireless signal technologies, there is a vast and increasing body of literature devoted
to studying 
key performance metrics derived from the signal spectrum. 
A significant thread of this research stems from modeling the positions of transmitters, receivers, or users as points of a random (typically Poisson) point process, and then computing quantities of interest using the tools of stochastic geometry; some key references are \cite{Andrews2010} \cite{Andrews2011} \cite{Baccelli2008} \cite{Haenggi2008} \cite{Renzo2013} and see also the recent works \cite{Renzo2016} \cite{George2016} and their references and discussion.

\subsection{The standard model}

The generally accepted ``standard" model in this setting is that discussed in \cite{Win2009} 
(going back to \cite{Baccelli1997} \cite{Brown2000})
where transmitters are placed according to a homogeneous Poisson process in the plane~$\IR^2$ 
with the receiver at the origin
(justifiable by thinking of the ``fixed" receiver location as being randomly chosen over the area of a large network),
and the signal strength at the receiver from a transmitter at location $x\in\IR^2$ is given by
$\ell(x) S_x$, where $\ell(x)$ is a non-increasing function representing the deterministic
propagation loss of the signal over distance, and the $S_x$ are i.i.d.\ positive random variables representing random
shadow-fading effects, e.g., 
from signals traveling through large obstacles (shadowing)
and same signal interactions (multi-path fading). 
If the Poisson transmitter placements are denoted by $\Xi\subset \IR^2\backslash\{0\}$, 
then the signal spectrum is just the collection of points $\Pi:=\{\ell(x) S_x\}_{x\in \Xi}$.

In the standard model, the signal spectrum is easy to understand since basic theory 
says $(x,S_x)_{x\in\Xi}$
is a Poisson point process on $\IR^2 \times \IR_{\geq0}$ and~$\Pi$ is  a deterministic
function of this point process, hence it is also a Poisson point process, now on $\IR_{\geq0}$,
with mean measure read from the density of~$\Xi$, 
the distribution of~$S_x$, and the function~$\ell$
 \cite[(2.8) of Proposition 2.9]{Keeler2014}
 and c.f.\ \cite[Lemma~1]{Blaszczyszyn2013}.
The importance of this result is that if you believe the standard model, then
the distribution of the signal spectrum only depends on the mean measure, which can be 
estimated from the empirical signal spectrum in a number of ways, see \cite{Reynaud-Bouret2003}
and references there.

\subsection{Universality results}

The standard model can be
generalized to allow for the transmitters
to be placed deterministically, such as (historically and unrealistically) on a grid,
or according to a (not necessarily Poisson) point process \cite{Miyoshi2014}.
In these cases, the arguments of the previous paragraph do not apply
and moreover analytic or numerical computations for quantities of interest
are not always feasible. 
Much work has gone into understanding the signal spectrum
for various choices of transmitter configurations, shadow-fading distributions, and propagation loss functions.
In general, 
features of the signal spectrum may depend on model details, but 
we can ask: are there some (weak) hypotheses on the parameters of the model
such that 
a few measurable quantities approximately determine the distribution
of signal strengths? 
For the generalized standard model this question has been answered in the affirmative by
\cite{Keeler2014} where it is shown that if the density of transmitters in the plane is asymptotically regular
and the shadow variables are large in mean but small in probability
-- typically referred to as a ``strong" shadowing regime \cite{Blaszczyszyn2015} -- then the spectrum of signal strengths can be well-approximated by a Poisson process on $\IR_{\geq0}$
with intensity read from the parameters of the model 
(this general result built on the work of \cite{Blaszczyszyn2013} where the same thing is shown for a particular family of shadow-fading variables and propagation loss functions). 
The importance of this result is that under the hypotheses above,
the signal spectrum behaves approximately as if the transmitters were placed according to a Poisson
process and so again, the distribution
of the signal spectrum is (approximately) determined by the mean measure which can 
be estimated from the empirical signal spectrum.

Having such universal results for the standard model is a positive development
(though there is still work to be done
in developing tests for determining when the results of \cite{Keeler2014} can be safely applied in practice),
but there is one serious issue with this story that stems
from the standard model itself: it is clearly unrealistic
to assume that the shadow-fading variables associated to different transmitters are independent \cite{Gudmundson1991} \cite{Baccelli2015}. 
If two transmitters are close to one another, or close to the same path to the origin, 
then there should be correlation between the shadowing effects for those transmitters;
see the very nice survey \cite{Szyszkowicz2010} and references there.

\subsection{Dependent shadowing}

Similar to the standard model, there has been much work around studying 
the signal spectrum for different choices of transmitter configurations, propagation loss functions,
and correlation schemes for different shadow-fading variables. 
However, as discussed in \cite{Szyszkowicz2010}, many of these schemes
have fundamental consistency issues and moreover, even where tractable analysis is possible, the behaviour of the signal spectrum under different models can vary considerably. 

Extrapolating just a bit beyond the summary judgement of \cite{Szyszkowicz2010},
we argue that the right models to consider in the current climate is that of the generalized
standard model with the added feature that 
the variables $(S_x)_{x\in\Xi}$ are values of a \emph{log-Gaussian field}
with correlation function $\rho$ defined on $\IR^2 \times \IR^2$. That is, $S_x = e^{a Z_{x}+b}$ where for any collection of points $x_1,\ldots, x_k$ of~$\Xi\subset\IR^2$, the vector $(Z_{x_1},\ldots, Z_{x_k})$ is centered multivariate Gaussian with covariance matrix
$(\rho(x_i, x_j))_{i,j=1}^k$ and~$a,b$ are some parameters. 
(This construction is well-defined if~$\rho$ is a symmetric and positive semidefinite function.)
Moreover,~$\rho(x,y)$ should have a reasonably fast rate of decay in the distance of~$x$ and~$y$ and ideally will have
an angular component accounting for the difference of the angles 
between the lines from~$x$ and~$y$ to the origin.
The qualities on~$\rho$ are direct recommendations from \cite{Klingenbrunn1999} \cite{Szyszkowicz2010}, 
and as observed by \cite{Gudmundson1991}, \cite{Catrein2008},
modeling the shadowing effects by a 
log-Gaussian field is empirically justified.

\subsection{Paper contribution}

We take a first step in generalizing the universality results 
of \cite{Keeler2014} to more realistic  models,  by extending 
the results of \cite{Blaszczyszyn2013} to the case where the shadowing variables 
are correlated.
 In particular, we show that 
if in the standard model 
\begin{itemize}
\item the transmitters are placed according to
either 1) a deterministic or random pattern with certain regularity conditions, or 2) a homogeneous Poisson process,
\item the propagation loss function decays like the norm to the power $-\beta$ for $\beta>2$, and 
\item the shadow-fading variables satisfy $S_x:=\exp(\sigma Z_x -\sigma^2/\beta)$,
where $\sigma>0$ and the $Z_x$ are generated from a Gaussian field with correlation
function with reasonable decay to zero and a technical condition given later
(note that the conditions are satisfied by most statistically tractable correlation functions appearing in the literature, reviewed below),
\end{itemize}
then the signal spectrum converges in distribution to 
a Poisson process with explicit mean measure. The result of \cite{Blaszczyszyn2013}
is precisely the case above with independent shadow variables.
We actually prove much stronger results in Theorems~\ref{thm:mainproc}
and~\ref{thm:randnumpd2}
where we provide
rates of convergence  for this limit theorem 
in total variation distance 
and a specialized point process metric.
The rate of convergence quantifies the quality of Poisson approximation 
and we also provide simulations for finite configurations
to indicate parameter values where the approximation is good.

There are a number of studies of wireless signal strengths for correlated shadow variables,
see references in  \cite{Szyszkowicz2010}, but general results of the type presented here are 
virtually unprecedented. Perhaps closest to our
limit results are those of \cite{Szyszkowicz2014} where a limit theorem 
is shown for the \emph{sum} of the signal spectrum from a 
finite (but large) collection of transmitters assuming a specific  
transmitter layout and shadowing correlation function 
that is much less general than our setting.

\medskip

\subsection{Layout of the paper}

\medskip

In the next section we define the model more precisely, 
state our convergence results in greater detail, and 
discuss their applicability, including a brief overview of relevant correlation functions
(a more significant survey is in Appendix~\ref{sec:updfuncs}).
In Section~\ref{sec:simmarg} we provide simulation results for a finite network
where transmitters are placed according to a Poisson process and a hexagonal grid
and then in Section~\ref{sec:futurework} we provide a discussion of results and future work. We precisely state and
prove our approximation results 
in Section~\ref{sec:approxproc}.
Finally in Appendix~\ref{sec:updfuncs} we survey correlation functions
relevant to our results, and Appendix~\ref{sec:proofppp} contains some
technical results, including the proof
of a new general Poisson point process approximation, Theorem~\ref{thmppp}.

\section{Model definition and convergence results}

We first discuss the generalization of the standard model we will use throughout the paper.
As is typically done, we actually study $N:=\{g(x)/S_x)\}_{x\in \Xi}:=\{1/(\ell(x) S_x)\}_{x\in \Xi}$, the
 spectrum of \emph{inverse} signal strengths
since there tend to be many weak signals which cause the approximating Poisson
process to have a singularity at zero.

\subsection{Correlated lognormal model}
\label{ssec:model}

It is convenient to identify countable sets $\{y_i \colon i \in \mcI\}$ of distinct points with the counting measures $\sum_{i \in \mcI} \delta_{y_i}$ where $\delta_y$ is the Dirac measure at $y$. Also we identify locally finite measures $M$ on $\pR = (0,\infty)$ with non-decreasing, right-continuous functions $\widetilde{M} \colon \pR \to \Rplus$ satisfying $\lim_{t \to 0} \widetilde{M}(t) = 0$, via $M((0,t])=\widetilde{M}(t)$, $t>0$. 

\begin{Setup}
\label{setup}
Let $\xi{\subset\IR^2/\{0\}}$ be a deterministic locally finite  
collection of points in $\IR^2$ representing the transmitter locations,
and $g(x)=h(\norm{x})=(K\norm{x})^{\beta}$ for some $K>0$ and $\beta>2$.
Write $\xi=\{x_i: \ i\in \mcI^{{\xi}} \}$ where $\mcI^\xi$ is a finite or countable index set.

Let $\{Z_x,  x\in\IR^2\}$ be a Gaussian field with correlation function $\rho$, standardized so that $\mean Z_x=0$ and $\var(Z_x)=1$. Suppose that $\rho$ is radially dominated by a non-increasing function $\tilde{\varrho} \colon \Rplus \to [0,1]$, i.e., $\rho(x,y) \leq \tilde\varrho(\norm{x-y})$ for all $x,y \in \R^2$, and that $\tilde\varrho(r) < 1$ for $r>0$.
Let $\sigma > 0$ and $S_x=\exp(\sigma Z_x -\sigma^2/\beta)$, the shadow
random variable associated to location~$x$.

Set $Y_{i}=g(x_{i})/S_{i}$, where we write $S_i=S_{x_i}$, and let $N:=N\t{\xi, \sigma}$ be the
signal spectrum generated by the collection $\{Y_i\}_{i\in\mcI^{{\xi}}}$, that is, $N=\sum_{i\in \mcI^{{\xi}}}\delta_{Y_{i}}$.
Let $p^{(x)}(t)=\pr(0<g(x)/S_x\le t)$, $x\in \IR^2$, $p_i(t)={p^{(x_i)}(t)}$, and $M(t):=M\t{\xi,\sigma}(t):=\sum_{i\in\mcI^{{\xi}}} p_i(t) = \mean N(t)$. 

If the transmitter locations are random, we write $\Xi:=\{X_i: \ i\in \mcI^{{\Xi}} \}$ in place of $\xi$ and assume that $\Xi$ is independent of the Gaussian field $\{Z_x,  x\in\IR^2\}$. Note that (the cardinality of) the set $\mcI^{{\Xi}}$ is also random, although for most situations it is enough to set $\mcI^{{\Xi}} = \NN$. We assume that $\Xi$ has a mean measure $\lambda$, which means that $\lambda(A) = \mean \Xi(A) < \infty$ for every bounded Borel set $A \subset \IR^2$. The signal spectrum $N=N\t{\Xi,\sigma}=\sum_{i\in \mcI^{{\Xi}}}\delta_{Y_{i}}$ is constructed in the same way as above, but based on $\Xi$ now instead of $\xi$. Note that $M(t):=M\t{\Xi,\sigma}(t):=\int_{\IR^2} \pr(g(x)/S_x\leq t) \lambda(dx) = \mean N(t)$; see Equation~\eqref{eq:meanofrandom} below.

For deterministic or random transmitters we assume that
$\rho$ and the bounding function $\tilde\varrho$ have the following properties.
\begin{itemize}
\item[P1.](Uniform positive definiteness, u.p.d.) For every $\eps > 0$ there exists a $\delta=\delta(\eps) >0$ such that 
for all $n \in \IN$, for all $\tx_1, \ldots, \tx_n \in \R^d$ with $\min_{i \neq j} \|\tx_i-\tx_j\| \geq \eps$, and for all $v \in \R^n$, we have
\be{
  \sum_{i,j=1}^{n} v_i v_j \rho(\tx_i,\tx_j) \geq \delta \|v\|^{2}.
} 
\item[P2.]There is an $R>0$ such that the function $r\mapsto r \tilde\varrho^2\bigl(R + \sqrt{3}R(r-1)\bigr)$ is non-increasing for $r\geq 1$.
\end{itemize}
\end{Setup}

Except for Properties~P1 and~P2, the setup is easily understood from the discussion in the introduction. 
Property~P2 is satisfied for any reasonable correlation function, and
essentially requires that $r\tilde\varrho(r)$ is non-increasing eventually, but
is stated in a specific way to simplify the proof of our upcoming results. 
The u.p.d.\ Property~P1 ensures that the spectral norm of the inverse of the covariance matrix induced by
the Gaussian field observed at any finite subset of points is uniformly upper bounded (by $1/\delta$ if the minimal distance is $\eps$). Our method requires that
we compare the distribution of $Z_x$ to that of $Z_x$ given the value of the field
at a finite collection of points and this is where the inverse covariance matrix
appears, see \eqref{eq:conditionalcumulants} in the proof of Theorem~\ref{thm:mainproc} and Lemma~\ref{lem:removingnp}. In terms of practical applicability,~P2
is satisfied for a wide choice of isotropic correlation functions, i.e., $\rho$ of the form $\rho(x,y) = \rho_0(\norm{x-y})$, such as any exponential, Mat\'{e}rn or Gaussian correlation function, as well as certain finite range correlation functions. It is also satisfied for certain
variations of such functions, e.g., 
$\rho(x,y) = \rho_0(\norm{A(x-y)})$ for $A$ a regular matrix, 
or convolutions and convex combinations of such functions.
 In all of these cases $\delta$ can be computed explicitly (at varying cost);
 see Appendix~\ref{sec:updfuncs} for further discussion.

Property~P1 also resonates well with avoiding numerical problems when applying statistical methods based on wireless network data. Many procedures concerned with statements about the shadowing random field $\{S_x, x \in \IR^2\}$, such as maximum likelihood estimation of its parameters or Kriging prediction of its values at unobserved locations involve inverting the correlation matrix $C = (\rho(x_i,x_j))_{1 \leq i,j \leq n}$ of the data $x_1, \ldots, x_n$. So for feasability of the procedure it is necessary that $C$ is invertible (i.e., $\rho$ is strictly positive definite). For numerical stability of the procedure, it is important that $C$ is well-conditioned, which (since $C$ is a correlation matrix) amounts essentially to saying that the smallest eigenvalue, i.e., the maximal $\delta$ in the definition of uniform positive definiteness, is not too close to zero.

Further, considering the list of correlation functions used previously in wireless network models given as Table~1 in \cite{Szyszkowicz2010}, after ruling out those functions that are not positive semidefinite (which do not produce feasible collections of shadowing variables), all of the isotropic correlation functions are u.p.d.\ with the possible exception of 
the powered exponential $\rho(x,y)=\exp(\norm{x-y}^\nu)$, $0<\nu\leq 2$, which is not well-studied
in the literature since it is similar to but less natural than the Mat\'{e}rn model
 because the parameter $\nu$ does not interpolate nicely, e.g., there is fundamentally different behaviour for $\nu = 2$ than for other $\nu$.
For correlation functions that incorporate both distance and angle,
the question of u.p.d.\ is not so easily answered and therefore we postpone
addressing it to later study.

\subsection{Convergence results}

We study the limit as $\sigma\to\infty$ of $N\t{\xi, \sigma}$ and $N\t{\Xi, \sigma}$ under Setup~\ref{setup}.
For our convergence results, we additionally assume that deterministic
transmitter placements~$\xi$ have the asymptotic homogeneity property:
for some $\kappa > 0$,
\ben{\label{eq:asymhom}
\xi(\bar{B}(0,r))/(\pi r^2) \to \kappa,
}
and that random placements $\Xi$ are homogeneous with intensity $\kappa$: for some $\kappa>0$
and all $A\subset\IR^2$,
\ben{\label{eq:randhom}
\mean\Xi(A)=\kappa |A|,
}
where $|A|$ denotes the (Lebesgue) area of $A$.

Under these assumptions, the
following result about the mean measure established in \cite{Blaszczyszyn2013} and \cite{Keeler2014}, remains obviously valid in our correlated setting.
\begin{Theorem}\label{thm:meanlim}
Let $\xi$, $N:=N\t{\xi, \sigma}$, and $M:=M\t{\xi, \sigma}$ be defined as in Setup~\ref{setup} with $\xi$ satisfying~\eq{eq:asymhom}. Then for each $t>0$,
\be{
\lim_{\sigma\to\infty} M\t{\xi,\sigma}(t)= \kappa \pi t^{2/\beta}/K^2.
} 
If we replace $\xi$ by a random $\Xi$ satisfying~\eq{eq:randhom}, then for every $\sigma>0$ and each $t>0$,
\be{
M(t)=\kappa \pi t^{2/\beta}/K^2.
}
\end{Theorem}

Before stating a convergence implication of our main results, we need some terminology
to define one of the classes of transmitter configurations we study.
Call a point process~$\Xi$ a \emph{hard core process} with (minimal) distance $\eps_*>0$ if $\pr(\inf_{i,j \in \mcI^{\Xi}, i \neq j} \norm{X_i-X_j} \geq \eps_*)=1$. We call a homogeneous point process $\Xi$ \emph{second-order stationary} if for any bounded Borel sets $A,B \subset \IR^2$ the ``second moments'' $\mean \bigl( \Xi(A) \Xi(B) \bigr)$ are finite and remain the same if we shift the whole point process by an arbitrary vector $x \in \IR^2$. Our convergence result for the hard core setting also requires a mild ``$\mathbf{B}^+_2$-mixing'' condition (existence of a reduced covariance measure with values in $[-\infty,\infty)$) described in detail around~\eqref{eq:defbgamma}. Intuitively $\mathbf{B}^+_2$-mixing says that the potential of a point of the process to excite further points at some distance $\geq r$ decreases to zero as $r \to \infty$.

To explain these conditions: having a hard core distance is intuitive from an engineering point of view
and the second order stationarity and mixing conditions are mild from a modeling perspective and are satisfied by almost all models currently used in the engineering literature. For example the Poisson process, many cluster point processes, and \emph{all} determinantal point process are $\mathbf{B}^+_2$-mixing, when used in their stationary variants. In any of these models we can introduce a hard core distance in a number of ways without destroying the mixing and stationarity properties; see the more detailed discussion after~\eqref{eq:defbgamma}.

To state a convergence implication of our main results, 
denote convergence in distribution of point processes
by $\inlawto$ (technically this is
weak convergence of point process distributions with respect to the vague topology on the space $\mfn$ of locally finite counting measures on $\IR$), and
note that the convergences $N \inlawto \Upsilon$ below imply convergence in distribution of any continuous statistical function of the point processes as well as joint convergence in distribution of vectors $\bigl( N((s_1,t_1]), \ldots, N((s_n,t_n]) \bigr)$ for any $n \in \NN$ and any $s_i,t_i \in \pR$ with $s_i < t_i$ for $1 \leq i \leq n$.

\begin{Theorem}\label{thm:mainintro}
Assume Setup~\ref{setup} and let $\Upsilon$ 
be a Poisson process on $(0,\infty)$ with mean measure
given by the function $L(t)= \kappa \pi t^{2/\beta}/K^2$.
\begin{enumerate}
\item[(i)]
Suppose that the deterministic transmitter configuration $\xi$ satisfies~\eq{eq:asymhom}
and is such that $\min_{x\in\xi} \norm{x}>0$
and $\inf_{x,y\in\xi} \norm{x-y}>0$.

If $\tilde\varrho(r)=\textrm{O}(r^{-(1+a)})$ for some $a>0$, then 
as $\sigma\to\infty$,
 \be{
N\t{\xi, \sigma} \inlawto \Upsilon.
}
 \item[(ii)] Suppose that the random transmitter configuration $\Xi$ is a
second-order stationary hard core process on $\IR^2$ with intensity~$\kappa$
that satisfies the $\mathbf{B}^+_2$-mixing condition.

If $\tilde\varrho(r)=\textrm{O}(r^{-(1+a)})$ for some $a>0$, then 
as $\sigma\to\infty$,
 \be{
N\t{\Xi, \sigma} \inlawto \Upsilon,
 }

\item[(iii)]
Suppose that the random transmitter configuration $\Xi$
is a homogeneous Poisson point process with intensity $\kappa$. 
If for some $c>0$, 
we have that the constant $\delta$ of {\rm P1} satisfies $\delta=\delta(\eps)=\Omega(\eps^c)$ as $\eps\to0$, and for some $a>0$ we have that  $\tilde\varrho(r)=\textrm{O}(e^{-ar})$ as $r\to\infty$,
 then as $\sigma\to\infty$,
 \be{
N\t{\Xi, \sigma} \inlawto \Upsilon.
 }
\end{enumerate}
\end{Theorem}

The theorem says that if~$\sigma$ is large, the correlation between
shadow-fading variables is only moderate, the correlation function
is u.p.d., and the transmitter placements have reasonable properties, then the spectrum of signal strengths will look Poisson. In fact
our quantitative approximation results in Section~\ref{sec:approxproc}
show that it is fine to approximate $N$ by a Poisson process with mean measure given
by $M$ (or more realistically, an empirically observed measure) rather than the limiting $L$ of Theorem~\ref{thm:meanlim} even in the deterministic case.

Our approximation theorems below
imply that the rates of convergence in $\sigma$ for
the statements of Theorem~\ref{thm:mainintro}
are polynomial, and in certain cases exponential.
The latter occurs for example in the hard core process case if the decay of the correlation function is exponential. 
Numerically, our error bounds are quite conservative
due to their generality, 
but they are important for the rates of convergence and because  
they indicate parameter relationships where the signal spectrum
will be close to Poisson.

To comment on the three different situations of the theorem, it is not clear what is the most appropriate model for transmitter locations \cite{Chiaraviglio2016} \cite{Lee2013} \cite{Zhou2015}. Placing transmitters according to a Poisson process is tractable and can be heuristically justified by thinking of the user positioned ``randomly" in a large network, but in fact this heuristic is better captured by Item~$(ii)$ which seems like the most useful model for thinking about 
robust network architecture. 
However, we state all three versions to show that the result and our methods are robust.

\subsection{Method of proof}
 
The convergence result follows from first deriving approximation bounds in a certain metric $\dtwo$ between point processes restricted to $[0,t]$, see Theorems~\ref{thm:mainproc}
and~\ref{thm:randnumpd2},
and then use these to show convergence of the restricted point process in $\dtwo$,
Theorems~\ref{thm:fixprocconv} and~\ref{thm:randprocconv}, and hence also in distribution. This then implies convergence in distribution of the unrestricted point processes.

To show the approximation bounds, our main tool is Stein's method for distributional approximation; introductions
to Stein's method from different perspectives are found in
\cite{Barbour1992} \cite{Chen2011} \cite{BarbourChen2005} \cite{Ross2011}.
Our needs do not quite fit into existing results for Poisson process approximation,
so we provide a new general approximation result using Stein's method, Theorem~\ref{thmppp}.

\section{Simulation results for marginal counts}\label{sec:simmarg}
Our simulations were run in {\tt R} \cite{R} (using the packages
listed  in the acknowledgments
below). We have strived to use parameters that are reasonably realistic in a mobile phone setting without fixating too much on a specific technology or scenario.

Our transmitters were placed in two configurations: at the
center of the cells of a hexagonal grid and according to 
a homogeneous Poisson process (resampled for each run of the simulation),
both on a disc with radius~\SI{30}{\km} and having average density $\kappa=\SI{5}{\per\square\km}$. In view of the data analyzed in \cite{Lee2013} and \cite{Zhou2015}, noting that the former counted base stations regardless of network operator and technology used, this choice is near the upper limit. For the deterministic propagation loss function we chose $\beta = 3.6$ and $K = \SI{4000}{\per\km}$, which is very similar to the parameters used in \cite{Blaszczyszyn2013}, Subsection B5; for the choice of $\beta$ see also Section~2.6 in \cite{Goldsmith2005}. A crucial quantity for the quality of approximation is the logarithmic standard deviation of shadowing $\sigma$. In applications this quantity is customarily expressed as $10 \sd(\log_{10}(S))$\! \si{\dB} and referred to as $\sigma_{\si{dB}}$, such that $\sigma = \sigma_{\si{dB}} \cdot \log(10)/10$. Typical values for $\sigma_{\si{dB}}$ from empirical studies range from 4 to \SI{13}{\dB}; see Section~2.7 in \cite{Goldsmith2005}. In \cite{Blaszczyszyn2013} values from 8 to \SI{12}{\dB} are considered. We use \SI{10}{\dB}, which translates to $\sigma = \log(10) \approx 2.30$. We choose for $\rho$ the exponential correlation function, i.e., $\rho(x,y)=\exp\{-\norm{x-y}/s\}$, which has been often used in the literature before. The scale parameter $s$ is known as \emph{decorrelation distance} in this context and often takes values between 50 and \SI{250}{\m}, corresponding to the order of magnitude of typical obstacles; see Subsection~VI.B in \cite{Szyszkowicz2010} and the references given there. We chose $s=0.1, 0.2, \SI{0.5}{\km}$, the last value to also include a setting where correlation is stronger than might be usual.

In the rest of the paper we implicitly normalize signal strengths to multiples of a constant transmitter power $P$, whereas the \emph{absolute} received signal strengths from a transmitter at location $x$ would be $(S_x/g(x)) P$. In view of the rather high transmitter density, we choose $P = \SI{40}{mW}$, which is a realistic power for a small cell, but too low for a macrocell.
However, with this choice we obtain realistic strongest signal strengths in the range of mostly \SI{-110} to \SI{-60}{dBm} (i.e., $10^{-11}$ to $10^{-6}$\hbit\si{mW}) at the receiver. Note that other choices of $P$ or $K$ will just result in scaling the received powers, but not in any other way alter the outcomes.

For the simulation study we generate $10^4$ realizations of the signal spectrum and compute the count statistics for the number of signals (reciprocals of the powers in \si{mW}) falling in the interval $[0,t]$, where $t=10^6, 10^7, 10^8$. This corresponds to counting signals with power at least $\SI{-60}{dBm}, \SI{-70}{dBm}, \SI{-80}{dBm}$, respectively, at the receiver.

The outcome is summarized in Table~\ref{table:PPPHexsim}.
The first row of the table 
records the theoretical numerically computed mean values~$M(t)$ 
for transmitters placed according to a Poisson process, while the second and third row
are the empirical mean and variance of $N(t)$ for this case. 
The fourth and fifth row of values are the empirical mean and variance for the hexagonal
transmitter arrangement. Figures~\ref{fig:CDFPPP} and~\ref{fig:CDFHex} show the 
empirical CDF of $N(t)$ for Poisson process and hexagonal transmitter placement, respectively, 
plotted against the appropriate Poisson CDF for $s=0.1,0.2,0.5$ (P-P plot).

\begin{table*}[h!]
\centering
\captionsetup{font=small}
{\scriptsize
\renewcommand{\arraystretch}{1.2}
\begin{tabular}{lrrrcrrrcrrr}\toprule
& \multicolumn{3}{c}{$s = 0.1$} & & \multicolumn{3}{c}{$s = 0.2$} & & \multicolumn{3}{c}{$s = 0.5$}\\
\cmidrule{2-4} \cmidrule{6-8} \cmidrule{10-12}
& $t=10^6$ & $t=10^7$ & $t=10^8$ && $t=10^6$ & $t=10^7$ & $t=10^8$ && $t=10^6$ & $t=10^7$ & $t=10^8$\\ \midrule
\underline{Poiss}\\
mean & 0.35 & 1.27 & 4.56 &&  0.35 & 1.27 & 4.56 && 0.35 & 1.27 & 4.56 \\
sim mean & 0.35 & 1.28 & 4.52 && 0.35 & 1.27 & 4.58 && 0.35 & 1.27 & 4.56\\
sim var & 0.38 & 1.37 &  4.87 && 0.41 &  1.60 &  6.20 && 0.49 & 2.33 & 10.62 \\
\underline{Hex}\\
sim mean & 0.35 & 1.25 & 4.53 && 0.35 & 1.24 & 4.48 && 0.35 &  1.25 & 4.52 \\
sim var & 0.27 & 0.83 &  2.84 && 0.28 &  0.93 &  3.50 && 0.33 & 1.50 & 7.53 \\
\bottomrule\\
\end{tabular}
}
\caption{Simulation results for the signal spectrum with
$K=\SI{4000}{\per\km}$, $\beta=3.6$, $\sigma=\log(10)$, $\rho(x,y)=\exp\{-\norm{x-y}/s\}$. Transmitters are placed on 
a disc of radius $\SI{30}{\km}$ according to a Poisson process and at the centers 
of the cells of a hexagonal tiling, both with average density $\SI{5}{\per\square\km}$ (see text for discussion of the realism of these parameters).
}
\label{table:PPPHexsim}
\end{table*}

In Figure~\ref{fig:CDFPPP} we notice, as expected, that with increasing scale $s$ of the correlation the approximation becomes worse. Note that for $s=0$ (no correlation) it would be clear by the Poisson process transformation theorem that the signal counts follow the exact Poisson distribution (and correspondingly also satisfy mean$\,=\,$variance with regard to Table~\ref{table:PPPHexsim}); see Lemma~1 in \cite{Blaszczyszyn2013}. For $s=0.1$ the approximation is still very close to exact. With increasing $s$ it deviates towards a more and more overdispersed distribution, as can be seen from the ``flipped S''-shape in the plots and from the high variances in Table~\ref{table:PPPHexsim}. Having an overdispersed distribution corresponds well to the intuition that under strongly positively correlated shadowing, we get clusters of signals in our interval $[0,t]$ and thus higher probabilities for both very low and very high counts. In such situations the signal spectrum for the given $\sigma$ may be better approximated by a Poisson cluster process or a more general Cox process; see \cite{Keeler2014}.

Turning now to Figure~\ref{fig:CDFHex}, we notice that for smaller $s$ the signal counts are underdispersed as seen from the ``proper S''-shape in the plots. Only for $s = 0.5$ we are in the situation of overdispersed signal counts again. The same is reflected in the relations of means and variances in Table~\ref{table:PPPHexsim}. The underdispersion of signals at smaller $s$ can be attributed to the regularity of the hexagonal grid and the fact that in view of the limit theorem the $\sigma$ considered is too small. This is partly compensated by the correlated shadowing, and as $s$ becomes larger, the shadowing effect takes over.
What is remarkable is that for correlated shadowing, there is an intermediate range of parameter settings (that are more or less realistic for a mobile network context) in which the Poisson approximation for the signal spectrum in the hexagonal configuration is actually more accurate than in the case of transmitters distributed according to a Poisson process.

We see from the above considerations that the difference between the mean and the variance read from~Table~\ref{table:PPPHexsim} is a useful proxy for the quality of Poisson approximation. Further analysis (not presented here) shows as expected that approximation gets better as $\sigma$ increases, but worse as $t$ increases beyond the values given in the table; see also \cite{Keeler2016}. 

\begin{figure}[h!]
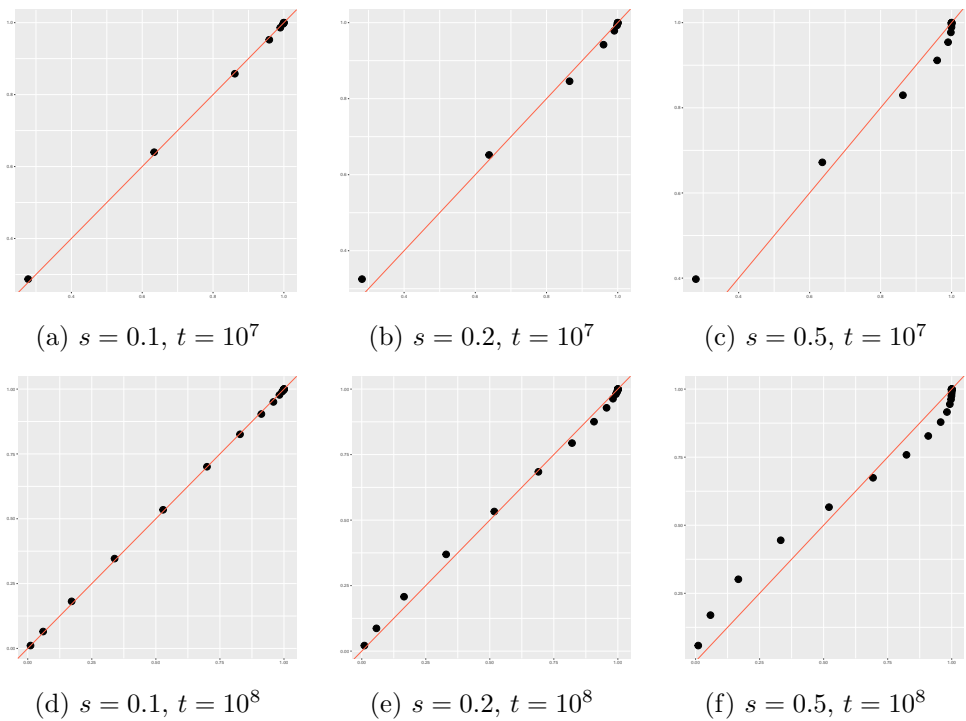

    \centering
    \captionsetup{font=small}
    \captionsetup[subfigure]{font=footnotesize,belowskip=5pt}
    \begin{subfigure}[h]{0.25\textwidth}
	\centering
	\includegraphics[width=\textwidth]{{{PPP_sigmalog10_beta3.6_density5_R30_s1tenth_K4000_P104_t107}}}
	\caption{$s=0.1$, $t=10^7$}    
	\end{subfigure}
    \ \ 
    \begin{subfigure}[h]{0.25\textwidth}
	\centering
	\includegraphics[width=\textwidth]{{{PPP_sigmalog10_beta3.6_density5_R30_s2tenth_K4000_P104_t107}}}
	\caption{$s=0.2$, $t=10^7$}
    \end{subfigure}
    \ \
    \begin{subfigure}[h]{0.25\textwidth}
	\centering
	\includegraphics[width=\textwidth]{{{PPP_sigmalog10_beta3.6_density5_R30_s5tenth_K4000_P104_t107}}}
	\caption{$s=0.5$, $t=10^7$}
    \end{subfigure}
     \\ 
     \mbox{}
          \\
      \begin{subfigure}[h]{0.25\textwidth}
	\centering
	\includegraphics[width=\textwidth]{{{PPP_sigmalog10_beta3.6_density5_R30_s1tenth_K4000_P104_t108}}}
\caption{$s=0.1$, $t=10^8$}
    \end{subfigure}
    \ \ 
    \begin{subfigure}[h]{0.25\textwidth}
	\centering
	\includegraphics[width=\textwidth]{{{PPP_sigmalog10_beta3.6_density5_R30_s2tenth_K4000_P104_t108}}}
	\caption{$s=0.2$, $t=10^8$}
    \end{subfigure}
    \ \
 \begin{subfigure}[h]{0.25\textwidth}
	\centering
	\includegraphics[width=\textwidth]{{{PPP_sigmalog10_beta3.6_density5_R30_s5tenth_K4000_P104_t108}}}
	\caption{$s=0.5$, $t=10^8$}
    \end{subfigure}    
    \caption{
    For transmitters placed  
    according to a Poisson process with indicated parameters,
    the black points have vertical coordinate the integer
    evaluations of the empirical CDF of the spectrum
    and horizontal coordinate the relevant Poisson CDF.
     }
    \label{fig:CDFPPP}
\end{figure}

\begin{figure}[h!]
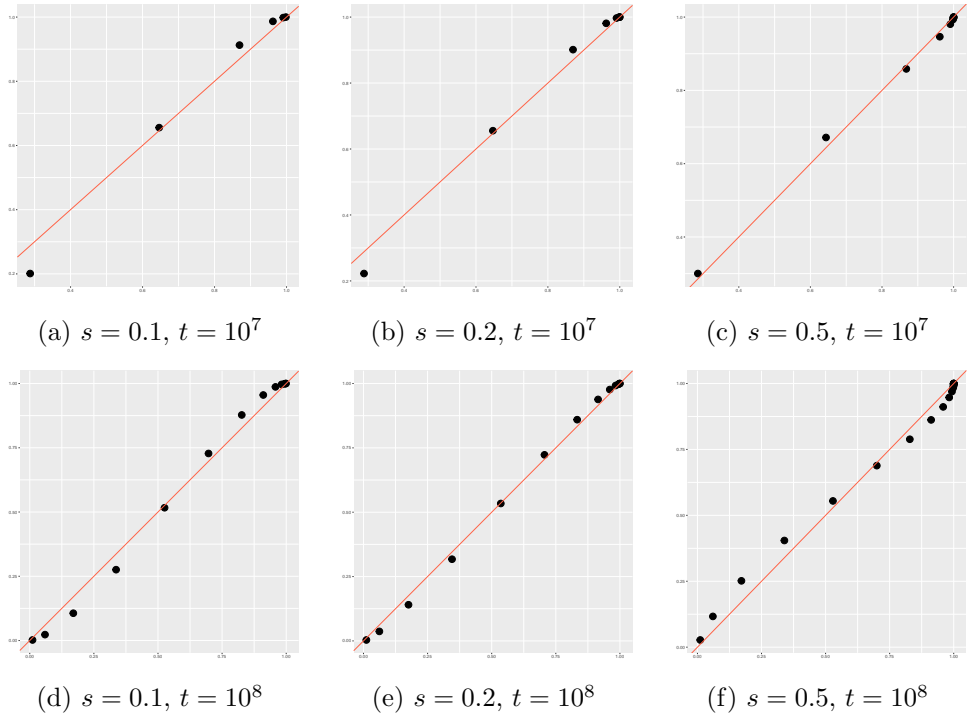

    \centering
    \captionsetup{font=small}
    \captionsetup[subfigure]{font=footnotesize,belowskip=5pt}
    \begin{subfigure}[h]{0.25\textwidth}
	\centering
	\includegraphics[width=\textwidth]{{{Hexp_sigmalog10_beta3.6_density5_R30_s1tenth_K4000_P104_t107}}}
	\caption{$s=0.1$, $t=10^7$}    
	\end{subfigure}
    \ \ 
    \begin{subfigure}[h]{0.25\textwidth}
	\centering
	\includegraphics[width=\textwidth]{{{Hexp_sigmalog10_beta3.6_density5_R30_s2tenth_K4000_P104_t107}}}
	\caption{$s=0.2$, $t=10^7$}
    \end{subfigure}
    \ \
    \begin{subfigure}[h]{0.25\textwidth}
	\centering
	\includegraphics[width=\textwidth]{{{Hexp_sigmalog10_beta3.6_density5_R30_s5tenth_K4000_P104_t107}}}
	\caption{$s=0.5$, $t=10^7$}
    \end{subfigure}
     \\ 
     \mbox{}
          \\
      \begin{subfigure}[h]{0.25\textwidth}
	\centering
	\includegraphics[width=\textwidth]{{{Hexp_sigmalog10_beta3.6_density5_R30_s1tenth_K4000_P104_t108}}}
\caption{$s=0.1$, $t=10^8$}
    \end{subfigure}
    \ \ 
    \begin{subfigure}[h]{0.25\textwidth}
	\centering
	\includegraphics[width=\textwidth]{{{Hexp_sigmalog10_beta3.6_density5_R30_s2tenth_K4000_P104_t108}}}
	\caption{$s=0.2$, $t=10^8$}
    \end{subfigure}
    \ \
 \begin{subfigure}[h]{0.25\textwidth}
	\centering
	\includegraphics[width=\textwidth]{{{Hexp_sigmalog10_beta3.6_density5_R30_s5tenth_K4000_P104_t108}}}
	\caption{$s=0.5$, $t=10^8$}
    \end{subfigure}    
    \caption{
    For transmitters placed  
    according to a hexagonal configuration with indicated parameters,
    the black points have vertical coordinate the integer
    evaluations of the empirical CDF of the spectrum
    and horizontal coordinate the relevant Poisson CDF.
     }
    \label{fig:CDFHex}
\end{figure}

\section{Discussion and future work}\label{sec:futurework}

We have shown that the spectrum of signal strengths in
a wireless network with moderately correlated strong shadowing 
is well approximated by a Poisson point process, generalizing 
the same result for the uncorrelated case. This holds true under 
a wide choice of probability distributions for the transmitter point process, in particular for virtually any second-order stationary distribution that includes a hard core, i.e.\ holds the transmitters at a minimal distance from each other.

Note that cellular networks would usually be designed to hold transmitters serving overlapping areas a certain distance apart. Alternatively such transmitters would operate at different frequencies to keep interference at a tolerable level, with frequencies being reused many times at transmitters that are somewhat farther apart; see \cite{Goldsmith2005}, Chapter~15. This means that even in networks where transmitters are strongly clustered, our theorem should still be valid if applied to the strengths of signals received at a certain frequency.

We have provided a first simulation study for our result using appropriate model parameters for a mobile phone setting. This study indicates that the Poisson process limit can become relevant for realistic correlation settings if $\sigma$ is reasonably high ($\sigma_{\si{dB}} = \SI{10}{dB}$ or above). It is an important avenue for future research to determine the applicability of our result more closely using data available from given networks, and to investigate its robustness across different scenarios. 

The importance of our result comes from the fact that in networks where it is applicable, we can essentially treat the signal spectrum as an inhomogeneous Poisson process, which is a point process that is well understood and relatively easy to handle. In particular, it is completely determined by its mean measure, which then can be estimated from the empirical spectrum in a number of ways. Also it has a number of important derived quantities that are mathematically tractable, such as coverage probabilities; see e.g.~\cite{Blaszczyszyn2015b} and the references in its introduction.

Some important questions related to our work that we leave open for 
future study are 
\begin{enumerate}
\item We expect our results to extend beyond the lognormal shadow-fading setting to 
shadow variables that are other functions of the value of a 
Gaussian field at the transmitter location in a similar way as 
\cite{Keeler2014} generalizes \cite{Blaszczyszyn2013} in the independent case.
\item We have not taken into account fast-fading
effects, which are frequently modelled by multiplying the shadow variable by 
another non-negative random variable, independent between transmitters. Apart from the fact that additional computations are necessary to adapt our proofs to the new distributions of shadowing and fading combined, we expect that the fast-fading effects would only help the Poisson convergence by adding a discontinuity in the correlation function at zero.  
\item To confidently apply our results in practice, it would be good
to develop simple statistical 
tests to determine when the hypotheses of strong shadowing under 
moderate correlation are satisfied; in particular to conclude when it is appropriate to assume the signal spectrum is approximately Poisson. Note that the analogous question
assuming independent shadowing is not
even well-addressed.
\item Many times we are interested in functions of the signal
spectrum such as the signal-to-interference ratio. Understanding
such statistics via our results would be useful.
\item Of the many empirical questions left untouched by our study,
it would be good to better understand which correlation functions and transmitter placements, 
or which properties of such
are most appropriate for modelling purposes.
\end{enumerate}

\section{Approximation results}\label{sec:approxproc}

In this section we bound distances between the distribution $\law(N \vert_{[0,t]})$ of our signal spectrum restricted to $[0,t]$ and a suitable Poisson process distribution in all of the three transmitter settings considered in Theorem~\ref{thm:mainintro}, i.e., deterministic, Poisson or hard core. Note that the results obtained here are much stronger than statements of convergence since they give concrete rates of 
convergence in terms of $\sigma$ and the other model parameters. The results also provide computable upper bounds, which, due to the generality and the technicality of our setting,
are typically too conservative to be of practical use.

The main metric we use for measuring distances between point process distributions is the Wasserstein metric $\dtwo$ with respect to the optimal subpattern assignment (OSPA) metric $\done$ between point configurations. We denote by $\mfn$ the space of finite point configurations $\psi = \{s_1,\ldots,s_n\} = \sum_{i=1}^n \delta_{s_i}$ (by the identification at the beginning of Subsection~\ref{ssec:model}) on $[0,t] \subset \IR$. (Note that all considerations about $\done$ and $\dtwo$ below remain valid for a compact subset of $\IR^d$ for arbitrary $d \in \NN$.)

The OSPA distance $\done(\psi,\upsilon)$ between point configurations $\psi, \upsilon \in \mfn$ can be roughly defined as the average Euclidean distance truncated at $1$ between points in an optimal pairing, where unpaired points (in the larger point configuration) count as paired at distance $1$. See~\cite{SVV08} for a precise definition (Equation~(3), where $p=c=1$) and a great deal of additional information. Since its introduction the OSPA metric has been abundantly used in signal processing and sometimes other disciplines of engineering, and has together with certain modifications become a standard evaluation tool for multi-object filtering and tracking algorithms; see \cite{Ristic2013}.

We can define the Wasserstein metric between distributions of point processes $\Psi$ and $\Upsilon$ on $[0,t]$ by coupling the two point processes in such a way that their expected $\done$-distance is as small as possible, more precisely
\be{
\dtwo \bigl( \law(\Psi),\law(\Upsilon) \bigr) = \min_{\Psi',\Upsilon'} \, \mean \hbit \done(\Psi',\Upsilon'),
}
where the minimum is taken over all pairs of point processes $\Psi'$, $\Upsilon'$ that have the same individual distributions as $\Psi$ and $\Upsilon$, respectively. This metric was studied in detail in~\cite{SX08}. Proposition~2.3(iii) in that article shows that $\dtwo$ describes the right topology for probability distributions of point processes in the sense that point processes $\Psi$, $\Psi_1, \Psi_2, \ldots$ satisfy 
\ben{
  \Psi_n \inlawto \Psi \qquad \text{if and only if} \qquad \dtwo(\law(\Psi_n), \law(\Psi)) \to 0
}  
as $n \to \infty$. Also in that article, Proposition~2.3(i) shows the equivalence of the following ``dual'' definition:
\ben{
  \dtwo \bigl( \law(\Psi), \law(\Upsilon)\bigr) = \sup_{f \in \mcf} \hbit \bigl| \E f(\Psi) - \E f(\Upsilon) \bigr|,  \label{eq:dtwoassup}
}
where 
\be{
  \mcf = \fw = \bigl\{ \tilde{f} \colon \mfn \to \IR \hbit;\, | \tilde{f}(\psi)-\tilde{f}(\upsilon) | \leq \done(\psi,\upsilon) \text{ for all $\psi, \upsilon \in \mfn$} \bigr\}
}
is the set of $1$-Lipschitz-continuous functions with respect to $\done$.

From this dual form one can see that upper bounds on the $\dtwo$-distance are upper bounds on $\bigl| \E f(\Psi) - \E f(\Upsilon) \bigr|$ for many useful point process statistics $f$, such as the average nearest neighbour distance between points and many $U$-statistics; see Section~3 in~\cite{SX08}. However, $f$ must be Lipschitz continuous and thus we do not directly get upper bounds for terms of the form
\be{ 
  \bigl| \pr(\Psi \in A) - \pr(\Upsilon \in A) \bigr|
} 
for general (measurable) sets $A \subset \mfn$. Therefore we present below also results in the total variation metric $\dtv$ given by
\be{
  \dtv(\law(\Psi),\law(\Upsilon)) = \sup_{A \in \mcn} \bigl| \pr(\Psi \in A) - \pr(\Upsilon \in A) \bigr| = \min_{\Psi',\Upsilon'} \, \pr(\Psi' \neq \Upsilon'),
}
where conveniently possible. The supremum above is taken over all measurable subsets of~$\mfn$ (technically, with respect to the Borel $\sigma$-algebra $\mcn$ generated by the vague topology on~$\mfn$). See~\cite[Appendix~A.1]{Barbour1992} for the second equality above and further results.

Since $\dtv$ can also be expressed as a supremum of the form \eqref{eq:dtwoassup} over the class $\mcf$ of \emph{all} measurable $[0,1]$-valued functions, which by $\done \leq 1$ contains in particular the class $\mcf_{W}$, we  obtain that $\dtwo \leq \dtv$. On the other hand, in general $\Psi_n \inlawto \Psi$ does \emph{not} imply $\dtv(\law(\Psi_n), \law(\Psi)) \to 0$, so $\dtv$ is a strictly stronger metric. Taking for $\Psi_n$ one deterministic point at $1/n$ and for $\Psi$ one deterministic point at $0$, we obtain $\dtv(\law(\Psi_n), \law(\Psi)) = 1 \not\to 0$ and see that the total variation metric is in some cases too strong to be useful.

We use the following general theorem about Poisson process approximation
which uses the same ideas as results already proved in the literature (e.g.,\cite[Section~10.2]{Barbour1992} or \cite[Theorem~2.1]{Schuhmacher2005a}), but is geared towards our application. A proof can be found
in Appendix~\ref{sec:proofppp}. Denote by $\| \cdot \|_{\mathrm{TV}}$ the total variation norm on the space of finite signed measures and $\pop(\Lambda)$ the law
of a Poisson process with mean measure $\Lambda$. For a measure $\lambda$ on $\R$, we use $\lambda\vert_t$ as short hand notation for $\lambda\vert_{[0,t]}$ and continue to use $\lambda(t)$ for $\lambda([0,t])$.
\begin{Theorem}\label{thmppp}
Let $\mcI$ be a finite index set, and let $\lambda_i$ be a probability distribution on $\Rplus$ for each $i \in \mcI$. Suppose that $Y_i \sim \lambda_i$ are real valued random variables and set $\Psi = \sum_{i \in \mcI} \delta_{Y_i}$ and $\lambda = \sum_{i \in \mcI} \lambda_i$. For every $i \in \mcI$, choose $A_i \subset \mcI$ with $i \in A_i$. Then, for any $t > 0$, setting $I_i = \I[Y_i \leq t]$, $p_i = \mean I_i = \lambda_i(t)$, $p_{ij} = \mean(I_i I_j)$,
\ban{
  \dtwo \bigl(&\law(\Psi\vert_{t}), \pop(\lambda\vert_{t})\bigr) \notag\\
&\leq \min \biggl( 1,\frac{1+2\log^{+}(\lambda(t))}{\lambda(t)} \biggr) \Biggl[\sum_{i\in\mcI, j\in A_i} p_i p_j +
\mathop{\sum_{i\in\mcI, j\in A_i}}_{i \neq j} p_{i j} 
+ \sum_{i\in \mcI} \mean \left|\mean \left[ I_i | \mcf_i \right] - p_i\right| \label{eq:pppdtwot1}\\
&\hspace*{54mm} {} +\sum_{i\in \mcI} \int_0^t \E \bigl| \pr(Y_i \leq s \mvert \mcf_i) - \pr(Y_i \leq s) \bigr| \; ds \Biggr],\label{eq:pppdtwo}
}
and
\ban{
  \dtv \bigl(&\law(\Psi\vert_{t}), \pop(\lambda\vert_{t})\bigr) \notag\\
&\leq \Biggl[\sum_{i\in\mcI, j\in A_i} p_i p_j +
\mathop{\sum_{i\in\mcI, j\in A_i}}_{i \neq j} p_{i j} \Biggr] + \sum_{i\in \mcI} \E \Bigl\| \law(Y_i \mvert \mcf_i)\big\vert_{t} - \law(Y_i)\big\vert_{t} \Bigr\|_{\mathrm{TV}}. \label{eq:pppdtv}
}
where $\mcF_i$ is any $\sigma$-field containing $\sigma(I_j, Y_j I_j: j\not\in A_i)$, in particular any $\sigma$-field containing $\sigma(Y_j: j\not\in A_i)$.
\end{Theorem}

\subsection{Deterministic transmitter placements}

Using Theorem~\ref{thmppp}, we first address the case of fixed transmitter placements. Denote by $\bar B(x,r)$ and $\mathring B(x,r)$ the closed and open balls of radius $r$ centered at $x$, respectively.
\begin{Theorem}\label{thm:mainproc}
Let $t>0$ and $\sigma, \xi$, and $N:=N\t{\xi, \sigma}$ be defined as in the Setup~\ref{setup}. Let $\Pi:=\Pi\t{\xi, \sigma}$ be a Poisson process on $\Rplus$ with the same mean measure~$M\t{\xi,\sigma}$ as~$N$.
Let $R>0$ be as in {\rm P2}
and define for $r>0$,
\be{
T_r(R)=\max_{x: \norm{x}\leq  r} \xi(\bar B(x,R)),
} 
where $\xi(\bar B(x,R))$  denotes the number of points in $\xi \cap \bar B(x,R)$. For $C \geq R$, 
let $d_*:=\min_{x\in\xi} \norm{x}>0$,
  $b_*:= \frac{1}{\sigma} \log \bigl( \tfrac{h(d_*)}{t} \bigr) + \frac{\sigma}{\beta}$,
$B_C:=\frac{1}{\sigma} \log \bigl( \tfrac{h(C)}{t} \bigr)+ \frac{\sigma}{\beta}$, 
$\eps_C = \inf \{\norm{x-y} : x,y \in \xi,\linebreak[3] 0 < \norm{x-y} < C \} > 0$, and $\delta_C$ be the  uniform positive definite constant of {\rm P1} for $\eps=\eps_C$. Furthermore, let
\be{
F=F(R,C)=\frac{1}{\delta_C} (4\pi+1) T_C(R) \biggl( \tilde\varrho^2(R) + \frac{1}{\sqrt{3}R^2} \int_R^\infty s \tilde\varrho^2(s) \; ds \biggr).
}
Then, requiring $\sigma^2>-\beta  \log \bigl( h(d_*)/t \bigr)$ (i.e., $b_*>0$), and also $B_C > 1$ and $B_C^2 F \leq 1$, 
\ba{
&\dtwo(\law(N\vert_t), \law(\Pi\vert_t)) \notag\\
&\quad \leq 2\int_{\norm{x}>C}\pr(t S\geq g(x)) \; \xi(dx) \notag\\
&\quad \quad+\min \biggl( M(t),1+2\log^{+}(M(t)) \biggr) \Biggl[ \, T_C(R)\bigl[ \pr(tS \geq h(d_*))
+ 5 e^{-b_*^2(1-\tilde\varrho(\eps_C))/4} \bigr]\notag\\
&\qquad\quad\quad+\frac{8(t+1)}{\sqrt{1-F^2}} \bigl( B_C + \sigma^{-1} \bigr) \sqrt{F} 
   + (t+1) (1+b_*^{-2})\sqrt{F} e^{-b_{*}^2(F^{-1}-1)/2} \Biggr].
   }
Requiring only $\sigma^2>-\beta  \log \bigl( h(d_*)/t \bigr)$,
\ba{
&\dtv(\law(N\vert_t), \law(\Pi\vert_t)) \notag\\
&\quad \leq 2\int_{\norm{x}>C}\pr(t S\geq g(x)) \; \xi(dx) + M(t)T_C(R)\bigl[ \pr(tS \geq h(d_*)) + 5 e^{-b_*^2(1-\tilde\varrho(\eps_C))/4} \bigr] \notag\\
&\quad\quad + \frac{4}{3} \Bigl( \pi \frac{C^2}{R^2} + (5 \pi+3) \frac{C}{R} \Bigr) T_C(R) (\sqrt{F} + 2F). 
}
\end{Theorem}

Before proving the theorem, we show how it implies that for $\sigma$ large
and moderate correlation, $N\vert_t$ is approximately Poisson.
\begin{Theorem}\label{thm:fixprocconv}
Let $t>0$ and $\sigma, \xi$, and $N:=N\t{\xi, \sigma}$ be defined as in the Setup~\ref{setup}. Let $\Pi:=\Pi\t{\xi, \sigma}$ be a Poisson process on $\Rplus$ with the same mean measure~$M\t{\xi,\sigma}$ as~$N$.
Assume $\xi$ is such that $\min_{x\in\xi} \norm{x}>0$
and $\eps:=\inf_{x,y\in\xi} \norm{x-y}>0$.
\begin{enumerate}
\item If $\tilde\varrho(r)=\textrm{O}(r^{-(1+a)})$ for some $a>0$, then 
as $\sigma\to\infty$,
 \be{
 \dtwo(\law(N\vert_t), \law(\Pi\vert_t)) \to 0.
 }
 \item If $\tilde\varrho(r)=\textrm{O}(r^{-(1+a)})$ for some $a>16/(1-\tilde\varrho(\eps))$,
 then 
  \be{
 \dtv(\law(N\vert_t), \law(\Pi\vert_t)) \to 0.
 }
 \end{enumerate}
\end{Theorem}
\begin{proof}
We apply the bounds of 
Theorem~\ref{thm:mainproc}. 
Set $C=\exp\{\sigma^2/\beta^2+\sigma^{1.1}\}$, 
let $d_*:=\min_{x\in\xi} \norm{x}>0$, 
$\eps_C=\eps:=\min_{x,y\in\xi} \norm{x-y}>0$,
and recall the notation of Theorem~\ref{thm:mainproc}.

For the first term appearing in both bounds, 
we claim $\int_{\norm{x}>C}\pr(t S\geq g(x)) \, \xi(dx) \to 0$, 
as long as
\ban{\label{eq:polart1}
\int_{r>C}\pr(t S\geq h(r)) \hbit r \; dr \to 0.
}
This is because
\ba{
\int_{\norm{x}>C}\pr(t S\geq g(x)) \; \xi(dx)&=\mean \int_{\norm{x}>C}\I[t S\geq g(x)] \; \xi(dx)\\
	&=\mean \: \xi \bigl( \bigl\{x: C < \norm{x} \leq \tfrac{S^{1/\beta} t^{1/\beta}}{K} \bigr\} \bigr).
}
Since $\xi(\bar B(0,r))/(\kappa \pi r^2)\to 1$, we have that for any $\iota>0$, 
if $C$ is large enough, then
we can bound this last term
\ba{
\mean \: \xi \bigl( &\bigl\{x: C < \norm{x}\leq \tfrac{S^{1/\beta} t^{1/\beta}}{K} \bigr\} \bigr) \\
&  \leq \kappa \pi  \mean \left[\left((1+\iota) 
\left(\frac{S^{1/\beta} t^{1/\beta}}{K}\right)^2-(1-\iota)C^2\right)\I[t S\geq h(C)]\right],
}
which, up to a constant factor (of $2\kappa\pi$) and the terms with the factors of $\iota$, can be rewritten as the left hand side of~\eq{eq:polart1}.
To take care of the $\iota$ terms, bound
\be{
\iota\mean\left(\frac{S^{1/\beta} t^{1/\beta}}{K}\right)^2\I[t S\geq h(C)]\leq 
\iota\mean\left(\frac{S^{1/\beta} t^{1/\beta}}{K}\right)^2=\iota \frac{t^{2/\beta}}{K^2},
}
using the moment generating function of a Gaussian random variable. Since we can take
$\iota\to0$ as $C\to\infty$, and the right hand side of the inequality is independent of $C$, the left hand side goes to zero as $\sigma\to\infty$.
For the other $\iota$ term, we use the usual Gaussian Mills ratio bound. If $r>0$, then
\ben{\label{eq:gaussmill}
\max\left\{\frac{1}{r+1},\frac{r^2}{r^2+1} \frac{1}{r} \right\}\leq \frac{\int_{r}^\infty e^{-u^2/2} du}{ e^{-r^2/2}} \leq \min\left\{\sqrt{\frac{\pi}{2}}, \frac{1}{r}\right\}.
}
Combining the upper bound with
\ben{\label{eq:bcgud}
\begin{split}
B_C&=\frac{\beta}{\sigma}\log(C) +\frac{\sigma}{\beta} +\Theta(\sigma^{-1}) \\
	&=\frac{2\sigma}{\beta} + \beta \sigma^{0.1}+\Theta(\sigma^{-1}),
\end{split}	
}
we find $\iota C^2\pr(t S\geq h(C))\to 0$ since
\be{
 \iota \exp\{2\sigma^2/\beta^2+2\sigma^{1.1}-(2\sigma/\beta + \beta \sigma^{0.1})^2/2\}\to0.
}
Now, to show~\eq{eq:polart1}, note that
the Mills ratio bound implies that 
it is enough to show 
\be{
\int_C^{\infty} \exp\{-(\beta/\sigma \log(r)+\sigma/\beta)^2)/2\}  r dr \to 0.
}
Now making the change of variable $u=(\beta/\sigma)\log(r)-\sigma/\beta$,
we find that the previous integral equals
\be{
\frac{\sigma}{\beta} \int_{\beta \sigma^{0.1}}^\infty e^{-u^2/2} \; du\to 0,
}
as desired. 

For the next term appearing in both bounds:
\ben{\label{eq:rest2}
T_C(R)\bigl[ \pr(tS \geq h(d_*)) + (5/2) e^{-b_*^2(1-\tilde\varrho(\eps_C))/4} \bigr],
}
note that because $\eps>0$, $T_C(R)=\mathrm{O}(R^2)$ not depending on $C$,
and because $d_*>0$ that $b_*=\sigma/\beta+\Theta(\sigma^{-1})$, so we find
\be{
\pr(tS \geq h(d_*))\leq e^{-b_*^2/2}\leq  e^{-b_*^2(1-\tilde\varrho(\eps))/4}
=\textrm{O}\left(e^{-\sigma^2(1-\tilde\varrho(\eps))/(4\beta^2)}\right).
}
For the convergence in $\dtwo$ we set $R=\sigma^{2/a}$ and for the
convergence in $\dtv$ we set $R=\exp\{2\sigma^2/(a\beta^2)+2\sigma^{1.1}/a\}$ and in both cases~\eq{eq:rest2}
tends to zero (using the inequality $b>16/(1-\tilde\varrho(\eps))$ for $\dtv$).

The tail condition on $\tilde\varrho$ implies $F=\textrm{O}(R^{-2a})$,
and so in both cases it is easy to see that the final remaining term
in the bound of Theorem~\ref{thm:mainproc} tends to zero.
\end{proof}

\begin{proof}[Proof of Theorem~\ref{thm:mainproc}]
We first truncate  in order to be able to work with a finite sum.
Let $\xi_C = \xi \vert_{\bar B(0,C)}$, $\mcI^{\xi_C} = \{ i \in \mcI^{\xi}; \, x_i \in \xi_C \}$, $N_C:=\sum_{i\in\mcI^{\xi_C}}\delta_{Y_i}$ and define $\Pi_C$ to be the Poisson process on $\pR$
having mean given by $M_C(t)=\mean N_C(t)$. 
 Writing $d$ for either $\dtwo$ or $\dtv$, we may split up the initial distance as
\ban{
d(\law(N\vert_t), \law(\Pi\vert_t))&\leq d(\law(N\vert_t), \law(N_C\vert_t)) +d(\law(\Pi_C\vert_t), \law(\Pi\vert_t))\label{eq:truncerr_proc} \\
	&\qquad +d(\law(N_C\vert_t), \law(\Pi_C\vert_t)). \label{eq:truncd_proc}
}
We can bound the first two summands by using direct couplings,
\ban{
\dtwo(\law(N\vert_t), \law(N_C\vert_t)) \leq \dtv(\law(N\vert_t), \law(N_C\vert_t))&\leq \pr(N\vert_t\neq N_C\vert_t) \notag\\	
	&= \pr(\exists i \in \mcI^{\xi} \setminus \mcI^{\xi_C} \colon tS_{x_i} \geq g(x_i)) \notag\\
  &\leq \int_{\norm{x}>C}\pr(t S\geq g(x)) \; \xi(dx) \label{eq:dircoup1}
}
(note that $\int f(x) \, \xi(dx) = \sum_{i\in \mcI^{\xi}} f(x_i)$ for measurable $f \colon \IR^2 \to \Rplus$), and
\ban{
\dtwo(\law(\Pi\vert_t), \law(\Pi_C\vert_t)) \leq \dtv(\law(\Pi\vert_t), \law(\Pi_C\vert_t))&\leq
\pr(\Pi\vert_t \neq \Pi_C\vert_t) \notag\\
        &= \pr \bigl(\Pi(t)-\Pi_C(t) > 0 \bigr) \notag\\
        &\leq M(t) - M_C(t) \notag\\
	&= \int_{\norm{x}>C}\pr(t S\geq g(x)) \; \xi(dx). \label{eq:dircoup2}
}

We apply Theorem~\ref{thmppp} to get an upper bound for the summand $d(\law(N_C\vert_t), \law(\Pi_C\vert_t))$. For $i \in \mcI^{\xi_C}$ ($\neq \emptyset$ w.l.o.g.), set $A_i=\{j \in \mcI^{\xi_C}: \norm{x_j-x_i}\leq R\}$ and $\mcf_i = \sigma(Z_{x_j}: j \in A_i^c)$, where $A_i^c = \mcI^{\xi_C} \setminus A_i$.
The first two terms appearing in~\eq{eq:pppdtwot1} and~\eq{eq:pppdtv} are the same 
up to prefactors. We bound them as follows.\\

\noindent\textbf{Term 1 of both.}
We find
\ben{\label{eq:pobdt1}
\begin{split}
\mathop{\sum_{i\in \mcI^{\xi_C}}}_{j\in A_i} p_i(t) p_j(t) &= \int_{\IR^2} \int_{\bar B(x,R)} \pr(t S_y \geq g(y))\pr(t S_x\geq g(x)) \; \xi_C(dy) \, \xi_C(dx) \\[-3mm]
	&\leq\int_{\IR^2}  T_{\norm{x}}(R)  \pr(t S\geq h(d_*))\pr(t S\geq g(x)) \; \xi_C(dx) \\[1mm]
	&\leq M(t) \hbit T_C(R) \hbit \pr(tS \geq h(d_*)).
\end{split}
}

\noindent\textbf{Term 2 of both.} We have
\ben{ \label{333}
\mathop{\sum_{i\in \mcI^{\xi_C}}}_{j\in A_i, j\neq i} p_{i j}
\leq 2 \int_{\IR^2}\mathop{\int_{\bar B(x,R) \setminus \{x\}}}_{\norm{y}\geq \norm{x}} \pr(t S_{x}\geq g(x), t S_{y}\geq g(y)) \; \xi_C(dy) \, \xi_C(dx).
}
Write
\ben{
b_x=b_x(\sigma):= \frac{1}{\sigma} \log \Bigl( \frac{g(x)}{t} \Bigr) + \frac{\sigma}{\beta} \label{eq:bxdef}
}
and note that we are assuming $\sigma$ is large enough so that 
$b_x>0$ for all $x\in\xi$ (i.e., $\beta  \log \bigl( h(d_*)/t \bigr)>-\sigma^2$).
We use the Gaussian Mills ratio bounds \eqref{eq:gaussmill}
to find
\ban{
\pr(t S_{x}\geq& g(x), t S_{y}\geq g(y)) = \pr(Z_x > b_x, Z_y > b_y) \notag \\
&\leq \frac{\pr(Z_x+Z_y > b_x+b_y)}{\pr(Z_x > b_x)} \pr(Z_x > b_x) \notag \\
&\leq \frac{\exp(-\frac{(b_x+b_y)^2}{4 (1+\rho(x,y))})}{\exp(-\frac{b_x^2}{2})} \min\left\{\frac{2b_x}{b_x+b_y} \frac{b_x^2+1}{b_x^2},
\sqrt{\frac{\pi}{2}} (1+b_x)\right\} \pr(Z_x > b_x), \label{35}
}
here we have used that $\rho(x,y)\leq 1$.
Now note that since $\norm{y}\geq \norm{x}$ and $g$ is non-decreasing in the norm of its argument, $b_x+b_y\geq 2b_x$,
so that~\eq{35} is bounded above by
\ben{\label{336}
\exp\left\{-\frac{b_x^2(1-\rho(x,y))}{2(1+\rho(x,y))}\right\} 
\min\left\{\frac{b_x^2+1}{b_x^2}, \sqrt{\frac{\pi}{2}} (1+b_x)  \right\}
\pr(Z_x > b_x).
}
Note that the minimum appearing in~\eq{336}
is of an increasing and decreasing function in $b_x$ and so is bounded by the maximum 
of the functions evaluated at a common point. Choosing $b_x=5/\sqrt{2\pi}-1$, we find that the minimum
appearing in~\eq{336}
is bounded by $5/2$ (numerically the minimum can be found to be around $2.34$).
To bound the remaining terms in~\eq{336},
we use that $b_x\geq b_*$
and that for $\norm{x}\leq C, \norm{y}\leq C$, $\rho(x,y)\leq \tilde\varrho(\eps_C) <1$, which implies
\ben{
\exp\left\{-\frac{b_x^2(1-\rho(x,y))}{2(1+\rho(x,y))}\right\}\min\left\{\frac{b_x^2+1}{b_x^2}, \sqrt{\frac{\pi}{2}} (1+b_x)  \right\}\leq \frac{5}{2} \exp\left\{-\frac{b_*^2(1-\tilde\varrho(\eps_C))}{4}\right\}.
\label{337}
}
Noting that
\be{
\int_{\IR^2} \int_{\bar B(x,R) \setminus \{x\}} \pr(Z_x>b_x) \; \xi_C(dy) \, \xi_C(dx) \leq T_C(R) M(t),
}
we obtain from~\eq{333} and~\eq{35}--\eq{337} that
\be{
  \mathop{\sum_{i\in \mcI^{\xi_C}}}_{j\in A_i, j\neq i} p_{i j} \leq M(t) \hbit T_C(R) \hbit  5 e^{-b_*^2(1-\tilde\varrho(\eps_C))/4}. 
}


\noindent\textbf{Term 3 of $\dtwo$.} 
We obtain
\ban{
\sum_{i\in \mcI^{\xi_C}} \mean \left|\mean \bigl[ I_i \mvert \mcF_i\bigr] - p_i\right| &= \sum_{i\in \mcI^{\xi_C}} \mean \left|\pr \bigl( t S_{x_i} \geq g(x_i) \bigm| (Z_{x_j})_{j \in A_i^c} \bigr) - \pr \bigl( t S_{x_i} \geq g(x_i) \bigr) \right| \nonumber \\
&= \sum_{i\in \mcI^{\xi_C}} \mean \left|\pr \bigl( Z_{x_i} \geq b_i \bigm| (Z_{x_j})_{j \in A_i^c} \bigr) - \pr \bigl( Z_{x_i} \geq b_i \bigr) \right|, \label{eq:weak1}
}
where 
\be{
b_i:=b_{x_i}= \frac{1}{\sigma} \log \Bigl(\frac{g(x_i)}{t} \Bigr) + \frac{\sigma}{\beta}.
}
Write $\widetilde Z_i = (Z_{x_j})_{j \in A_i^c}$, which is interpreted as a column vector of length $n_i := \# A_i^c$.
Since $Z_{x_i}$, $\widetilde Z_i$ are jointly $(1+n_i)$-variate normally distributed with mean vector $0 \in \R^{1+n_i}$ and $(1+n_i) \times (1+n_i)$ covariance matrix
\be{
  \left(\begin{array}{cc}
    1      & \gamma^{\top} \\[2mm]
    \gamma & \Gamma        \\
  \end{array}\right),
} 
where
\be{
  \gamma = \bigl( \rho(x_i,x_j) \bigr)_{j \in A_i^c}, \text{ and } \ \Gamma = \bigl( \rho(x_j,x_k) \bigr)_{j,k \in A_i^c},
}
we obtain by a standard result that $\mathcal{L} \bigl( Z_{x_i} \bigm| (Z_{x_j})_{j \in A_i^c} = \tilde{z} \bigr) = \mathcal{N}(\mu_i(\tilde z), \tau_i^2)$, where
\ben{ \label{eq:conditionalcumulants}
  \mu_i(\tilde{z}) = \gamma^{\top} \Gamma^{-1} \tilde{z} \text{ and } \ \tau_i^2 = 1 - \gamma^{\top} \Gamma^{-1} \gamma.
}
Note that $\Gamma$ is invertible because it is positive definite by Property~P2. Also $\tau_i^2 \in (0,1]$, because $\Gamma^{-1}$ is also positive definite
and $\tau_i^2$ is the conditional normal variance, which by $\rho(x,y) < 1$ for $x \neq y$ may not be zero.
 
Thus, from~\eqref{eq:weak1}, for a standard normally distributed random variable $Z$ that is independent from $\tZ_i$, we have
\ban{
\sum_{i\in \mcI^{\xi_C}} \mean \left|\mean \bigl[ I_i \mvert \mcF_i\bigr] - p_i\right| 
&= \sum_{i\in \mcI^{\xi_C}} \mean \left|\pr_{\tZ_i} \bigl( Z > \tfrac{b_i-\mu_i(\tZ_i)}{\tau_i} \bigr) - \pr \bigl( Z > b_i \bigr) \right|,\label{eq:condcumul2}
}
where $\pr_{\tZ_i}$ denotes the conditional probability given $\tZ_i$. 
Write 
$\tb_i = \tb_i(\tZ_i)=\frac{b_i-\mu_i(\tZ_i)}{\tau_i}$, and 
recall that $\sigma$ is large enough so that $b_i > 0$.
We then have
\ban{
&\left|\pr_{\tZ_i} \bigl( Z > \tfrac{b_i-\mu_i(\tZ_i)}{\tau_i} \bigr) - \pr \bigl( Z > b_i \bigr) \right|\notag \\ 
&\hspace{2cm}\leq \frac{e^{-b_i^2/2}}{\sqrt{2\pi}} (\tb_i-b_i) \I[\tb_i>b_i]+ \frac{e^{-\tb_i^2/2}}{\sqrt{2\pi}} (b_i-\tb_i) \I[0\leq \tb_i<b_i]
+\I[\tb_i<0] \notag \\
&\hspace{2cm} \leq \frac{e^{-b_i^2/2}}{\sqrt{2\pi}} (\tb_i-b_i) \I[\tb_i>b_i](1+e^{-(\tb_i^2-b_i^2)/2})
+ \frac{e^{-\tb_i^2/2}}{\sqrt{2\pi}} (b_i-\tb_i) 
+\I[\tb_i<0] \notag \\
&\hspace{2cm} \leq \frac{2 e^{-b_i^2/2}}{\sqrt{2\pi}} (\tb_i-b_i) \I[\tb_i>b_i]
+ \frac{e^{-\tb_i^2/2}}{\sqrt{2\pi}} (b_i-\tb_i) 
+\I[\tb_i<0].\label{eq:bd21}
}
We take the expectation of~\eq{eq:bd21} against $\tZ_i$, using that  $\tb_i$ is normal with mean $b_i/\tau_i$ and variance
$s_i^2/\tau_i^2$, where $s_i^2:=\gamma^{\top} \Gamma^{-1}\gamma=1-\tau_i^2$. The following Gaussian
expectation formulas can be checked by straightforward calculation: if $X$ is normal mean $m$ and variance $v^2$, then
\ba{
\mean e^{-X^2/2} &=\frac{1}{v\sqrt{2 \pi}} \int_{-\infty}^\infty e^{-u^2/2} e^{-(u-m)^2/(2v^2)} \; du= (v^2+1)^{-1/2} e^{-\frac{m^2}{2(1+v^2)}}, \\
\mean  X e^{-X^2/2} &=\frac{1}{v\sqrt{2 \pi}} \int_{-\infty}^\infty ue^{-u^2/2} e^{-(u-m)^2/(2v^2)} \; du=m (v^2+1)^{-3/2}  e^{-\frac{m^2}{2(1+v^2)}},  \\[0.5mm]
\mean X \I[X>0] &\leq \sqrt{v^2+m^2}.
}
We then find that the expectation of~\eq{eq:bd21} against $\tZ_i$ is bounded above by
\ba{
\frac{e^{-b_i^2/2}}{\sqrt{2\pi}}&\left[\frac{2}{\tau_i}\sqrt{s_i^2+b_i^2(1-\tau_i)^2}+\frac{\tau_i b_i(s_i^2+\tau_i(\tau_i-1))}{(s_i^2+\tau_i^2)^{3/2}}
e^{\frac{b_i^2(s_i^2+\tau_i^2-1)}{2(\tau_i^2+s_i^2)}}\right] +\pr\left(Z>\frac{b_i}{s_i}\right) \\
&\leq \frac{e^{-b_i^2/2}}{\sqrt{2\pi}}\left[\frac{2}{\tau_i}\sqrt{s_i^2(b_i^2 s_i^2+1)}+\tau_i b_i \bigl(s_i^2-\tau_i(1-\tau_i)\bigr) \right] +\pr\left(Z>\frac{b_i}{s_i}\right),
}
since $s_i^2+\tau_i^2=1$ and $(1-\tau_i) \leq (1-\tau_i^2) = s_i^2$.
Now combining this bound with the fact that 
$b_*\leq b_i\leq B_C$, the Mills ratio inequalities
of~\eq{eq:gaussmill} to find
\ba{
\frac{e^{-b_i^2/2}}{\sqrt{2\pi}}&\leq (b_i+1) \pr(Z>b_i)\leq  (B_C+1) \pr(Z>b_i),\\
\pr\left(Z> b_i/s_i\right)&\leq s_i (1+b_*^{-2})e^{-b_*^2(s_i^{-2}-1)/2} \pr(Z>b_i),
}
and that from~\eq{eq:removingn1p} of Lemma~\ref{lem:removingnp},
\ba{
0\leq s_i^2&\leq \frac{4\pi+1}{\delta_C}  T_C(R) \biggl( \tilde\varrho^2(R) + \frac{1}{\sqrt{3} R^2} \int_R^\infty s \tilde\varrho^2(s) \; ds \biggr) =:F(R,C)=F,
}
we  find that~\eq{eq:condcumul2}
is bounded above by
\be{
\begin{split}
& M(t)(B_C+1) \left[\frac{2}{\tau(R,C)} \sqrt{F(B_C^2F+1)}+B_CF \right] \\
&\hspace{4cm}+M(t)(1+b_*^{-2})\sqrt{F} \exp \bigl( -\tfrac12 b_{*}^2(F^{-1}-1) \bigr),
\end{split}
}
where $\tau(R,C)=\min_i \tau_i$. 
Note that because $1 < B_C$ and $B_C^2 F \leq 1$,
\be{
\begin{split}
  (B_C+1) \left[\frac{2}{\tau(R,C)} \sqrt{F(B_C^2F+1)}+B_CF \right]  & \leq \frac{4\sqrt{2}}{\tau(R,C)} B_C \sqrt{F} + 2 B_C^2 F \\
  	&\leq \frac{8}{\tau(R,C)} B_C \sqrt{F} \\
        &\leq \frac{8}{\sqrt{1-F^2}} B_C \sqrt{F}
\end{split}
}
by \eq{eq:removingn1p}, since by the above conditions also $F<1$.

Thus we obtain for \eqref{eq:weak1} the total bound of
\ben{
  M(t) \biggl( \frac{8}{\sqrt{1-F^2}} B_C \sqrt{F} + (1+b_*^{-2})\sqrt{F} e^{-b_{*}^2(F^{-1}-1)/2} \biggr) \label{eq:totbound3rd}
}

\noindent\textbf{Term 4 of $\dtwo$.}
To bound the final term~\eq{eq:pppdtwo}, we only 
have to integrate the bound~\eqref{eq:totbound3rd} coming from the third term of~\eqref{eq:pppdtwot1}. Writing explicitly the dependence of $B_C$ and $b_*$ on $t$, replacing $t$ by $s$ and integrating, 
noting that $F$ is independent of $s$ and $M(s)$ is non-decreasing in $s$, we obtain
\ban{
  \sum_{i\in \mcI^{\xi_C}} &\int_0^t \E \bigl| \pr(Y_i \leq s \mvert \mcf_i) - \pr(Y_i \leq s) \bigr| \; ds \notag\\
  &\leq M(t) \biggl( \frac{8}{\sqrt{1-F^2}} \sqrt{F} \int_0^t B_C(s) \; ds + \int_0^t (1+b_*(s)^{-2})\sqrt{F} e^{ -b_{*}^2(s)(F^{-1}-1)/2} \; ds \biggr). \label{eq:intwdp}
}
Since
\ba{
  \int_0^t B_C(s) \; ds &= \frac{1}{\sigma} \int_0^t \log \bigl( \tfrac{h(C)}{s} \bigr) \; ds + t \frac{\sigma}{\beta} \\
  &= \frac{1}{\sigma} \Bigl[ s + s \log \bigl( \tfrac{h(C)}{s} \bigr) \Bigr]_{s=0}^t + t \frac{\sigma}{\beta} \\
  &= \frac{t}{\sigma} + t B_C(t),
}
and $b_*$ is non-increasing in $s$, we can upper bound~\eqref{eq:intwdp} by
\ben{
  M(t) \biggl( \frac{8t}{\sqrt{1-F^2}} \bigl( B_C(t) + \sigma^{-1} \bigr) \sqrt{F} 
  + t (1+b_*(t)^{-2})\sqrt{F} e^{-b_{*}^2(t)(F^{-1}-1)/2} \biggr). \label{eq:intwdpout}
}

\noindent\textbf{Term 3 of $\dtv$.}
Turning to $d = \dtv$ now, it remains to estimate the term
\ben{
  \sum_{i\in \mcI^{\xi_C}} \E \Bigl\| \law(Y_i \mvert \mcf_i)\big\vert_{t} - \law(Y_i)\big\vert_{t} \Bigr\|_{\mathrm{TV}}
  = \sum_{i\in \mcI^{\xi_C}} \E \Bigl\| \law(Y_i \mvert (Z_{x_j})_{j \in A_i^c})\big\vert_{t} - \law(Y_i)\big\vert_{t} \Bigr\|_{\mathrm{TV}}  \label{eq:tvnorm}
}
in~\eqref{eq:pppdtv}. Note that for a finite signed measure $\mu$ and a bimeasurable bijection $g \colon \R \to \im(g) \subset \R$ we have
\be{
  \tvnorm{\mu g^{-1}} = \sup_{B \in \mcb} \mu(g^{-1}(B)) - \inf_{B \in \mcb} \mu(g^{-1}(B)) = \sup_{B \in \mcb} \mu(B) - \inf_{B \in \mcb} \mu(B) = \tvnorm{\mu}.
} 
Therefore
\ban{
  \Bigl\| \law(Y_i \mvert (Z_{x_j})_{j \in A_i^c})\big\vert_{t} - \law(Y_i)\big\vert_{t} \Bigr\|_{\mathrm{TV}}
  &= \Bigl\| \law(Z_{x_i} \mvert (Z_{x_j})_{j \in A_i^c})\big\vert_{t} - \law(Z_{x_i})\big\vert_{t} \Bigr\|_{\mathrm{TV}} \notag\\
  &\leq 2 \hbit \dtv \bigl( \mcn(\mu_i(\tZ), \tau_i^2), \mcn(0,1) \bigl) \notag\\
  &\leq 4 \bigl| 1-\tau_i^2 \bigr| + \sqrt{2 \pi} |\mu_i(\tZ)|,
 \label{eq:normdensdiff}
}
where the last inequality follows from Lemma~\ref{lem:dtvnormal} below.

Note that $1-\tau_i^2 = s_i^2 \leq F$ and $\mu_i(\tZ) \sim \mcn(0,s_i^2)$, so that $\E|\mu_i(\tZ)| = s_i \sqrt{2/\pi}$.
We bound the number of points of $\xi_C$ in $\bar{B}(0,C)$ by covering $\bar{B}(0,C)$ by $\bar{B}(0,R)$ and the annuli $B_k = \{ x \in \R^2 \colon R + (k-1) \sqrt{3} R < \norm{x} \leq R + k \sqrt{3} R \}$, $1 \leq k \leq \lceil (C-R)/(\sqrt{3} R) \rceil =: A$. By Lemma~\ref{lem:annuluscount} below, $B_k$ contains no more than $\lceil 4 \pi k \rceil T_C(R)$ points of $\xi_C$. Therefore
\ba{
  \# \mcI^{\xi_C} / T_C(R) &\leq 1 + \sum_{k=1}^A (4 \pi k + 1) \notag\\
  &= 2 \pi A^2 + (2 \pi +1) A + 1 \notag\\
  &\leq \frac{2\pi}{3} \frac{C^2}{R^2} + \frac{10 \pi+6}{ 3} \frac{C}{R},
}
since $C \geq R$. Thus we can bound \eqref{eq:tvnorm} by
\be{
  \Bigl( \frac{2\pi}{3} \frac{C^2}{R^2} + \frac{10 \pi+6}{ 3} \frac{C}{R} \Bigr) T_C(R) (4F+2\sqrt{F}).\qedhere
}
\end{proof}

The following lemma and its use in the proof of Theorem~\ref{thm:mainproc} above
clarifies the form of Property~P2 of~$\rho$ in Setup~\ref{setup}.

\begin{Lemma}\label{lem:removingnp}
Let $\rho$ be a correlation function radially dominated by $\tilde\varrho \colon \Rplus \to [0,1]$ and satisfying {\rm P1} and {\rm P2} above. Let $x_0 \in \IR^2$ and $\eps > 0$ be fixed, $\delta = \delta(\eps) >0$ as in Property~P1, and $R>0$ 
large enough to apply Property~P2. For arbitrary $n \in \NN$ let $x_1,\ldots,x_n \in {\bar B(x,R)}^{c}$ with $\min_{1 \leq i,j \leq n} \norm{x_i-x_j} \geq \eps$. Define
$T(R)$ to be the maximum number of the $x_i$ in a ball of radius $R$,
\be{
  \gamma = \bigl( \rho(x_0,x_j) \bigr)_{1 \leq j \leq n}, \text{ and } \ \Gamma = \bigl( \rho(x_j,x_k) \bigr)_{1 \leq j,k \leq n}.
}
Then
\begin{equation} \label{eq:removingn1p}
  0\leq \gamma^{\top} \Gamma^{-1} \gamma \leq \frac{1}{\delta} (4\pi+1) T(R) \biggl( \tilde\varrho^2(R) + \frac{1}{\sqrt{3} R^2} \int_R^\infty s \tilde\varrho^2(s) \; ds \biggr).
\end{equation}
\end{Lemma}
\begin{proof}
Note that the uniform positive definiteness of $\rho$ yields $\delta>0$ as a lower bound on the 
smallest eigenvalue of the correlation matrix $\Gamma$ (cf.\ the proof of Theorem~\ref{thm:upd} in the appendix). Hence the spectral norm of $\Gamma^{-1/2}$ is bounded above by $1/\sqrt{\delta}$. Thus
\be{ 
  \| \Gamma^{-1/2} \gamma \| \leq \frac{1}{\sqrt{\delta}} \| \gamma \|.
 }

To bound the norm of $\gamma$ we subdivide ${\bar B(x_0,R)}^c$ into the annuli $B_k = \{ x \in \R^2 \colon R + (k-1) \sqrt{3} R < \| x-x_0 \| \leq R + k \sqrt{3} R\}$, $k \in \NN$. By Lemma~\ref{lem:annuluscount} the annulus $B_k$ contains no more than $\lceil 4\pi k \rceil T(R)$ points of~$x_1,\ldots,x_n$. Therefore
\ba{
  \|\gamma\|^2 &= \sum_{j=1}^n \gamma_j^2= \sum_{k=1}^{\infty} \, \sum_{j \hbit : \hbit x_j \in B_k} \gamma_j^2 \nonumber\\
  &\leq \sum_{k=1}^{\infty} \lceil 4\pi k \rceil T(R) \tilde\varrho^2\bigl((1 + (k-1) \sqrt{3})R\bigr) \nonumber\\
  &\leq (4\pi+1) T(R) \sum_{k=1}^{\infty} k \tilde\varrho^2\bigl(R + \sqrt{3}R(k-1)\bigr) \nonumber\\
  &\leq (4\pi+1) T(R) \biggl( \tilde\varrho^2(R) + \int_1^\infty r \tilde\varrho^2\bigl(R + \sqrt{3}R(r-1)\bigr) \; dr \biggr) \nonumber\\
  &= (4\pi+1) T(R) \biggl( \tilde\varrho^2(R) + \frac{1}{\sqrt{3} R} \int_R^\infty \Bigl(\frac{s-R}{\sqrt{3} R}+1\Bigr) \tilde\varrho^2(s) \; ds \biggr) \nonumber\\
  &\leq (4\pi+1) T(R) \biggl( \tilde\varrho^2(R) + \frac{1}{\sqrt{3} R^2} \int_R^\infty s \tilde\varrho^2(s) \; ds \biggr),
}
where the third inequality holds because of {\rm P2}.
\end{proof}

The following lemma is a technical result used in
the proof of Theorem~\ref{thm:mainproc} which bounds the total variation distance
between two normal distributions.
\begin{Lemma}\label{lem:dtvnormal}
For $m\in \IR $ and $s>0$,
\be{
\dtv(\mcn(m,s^2),\mcn(0,1)) \leq  2 | 1-s^2| + \sqrt{\pi/2} \, |m|
}
\end{Lemma}
\begin{proof}
According to \cite[(2.12) of Lemma~2.4]{Chen2011}, 
if $W$ is a random variable and $Z\sim \mcn(0,1)$, then 
for any Borel set $A$,
\ben{\label{eq:dtvstn}
|\pr(W\in A)-\pr(Z\in A)| = |\mean[ f'(W)-W f(W) ]|,
}
where $f=f_A$ satisfies $|f(z)|\leq \sqrt{\pi/2}$ and $|f'(z)| \leq 2$. Stein's lemma
(or direct computation) 
says that
if $W\sim \mcn(m,s^2)$, then for any bounded $f$ with bounded derivative we have
\be{
s^2 \mean f'(W) - \mean (W-m) f(W) = 0.  
}
Subtracting this inside of the absolute value of~\eq{eq:dtvstn}, we have 
\ba{
|\pr(W\in A)-\pr(Z\in A)| & = |\mean[ f'(W)(1-s^2) - m f(W) ]| \\
	&\leq  | 1-s^2| \mean |f'(W)| + |m| \, \mean |f(W)|,\\
	&\leq 2 | 1-s^2| + \sqrt{\pi/2} \, |m|,
}
where we used the bounds on $f,f'$. Since this holds uniformly 
in~$A$, the lemma follows.
\end{proof}

The following elementary lemma used in the proofs of Theorem~\ref{thm:mainproc} and Lemma~\ref{lem:removingnp} allows to control the maximal number of points of $\xi$ in a set by subdividing it into annuli. 
\begin{Lemma}\label{lem:annuluscount}
Let $\eta \subset \IR^2$ be a locally finite set and $T(R) = \max_{x \in \IR^2} \eta(\bar{B}(x,R))$. Then for any $x_0 \in \IR^2$, the annulus $B_k = \{ x \in \R^2 \colon R + (k-1) \sqrt{3} R < \| x-x_0 \| \leq R + k \sqrt{3} R\}$ contains no more than $\lceil 4 \pi k \rceil T(R)$ points of $\eta$. 
\end{Lemma}
\begin{proof}
We may clearly set $x_0=0$ without loss of generality. Note that the annulus $B = \{ x \in \R^2 \colon r_0 - \frac{\sqrt{3}}{2} R < \| x \| \leq r_0 + \frac{\sqrt{3}}{2} R\}$ can be completely covered by $\bigl\lceil 2\pi \big/ \psi \bigr\rceil$ closed balls of radius $R$ with centers placed at $(\tilde r \cos(m \psi), (\tilde r \sin(m \psi))$, $0 \leq m \leq \bigl\lceil 2\pi \big/ \psi \bigr\rceil-1$, where $\tilde r = (r_0^2 + (R/2)^2)^{1/2}$ and $\psi = 2 \arctan(R/(2r_0))$. Therefore $B_k$ can be covered with $\bigl\lceil 2\pi \big/ \psi \bigr\rceil$ closed $R$-balls, where $r_0 = R+(k-\frac{1}{2})\sqrt{3}R$. By the fact that $\arctan$ is concave on $[0,\infty)$ with $\arctan(0)=0$, we obtain for any $a,b>0$ that
\begin{equation*}
  \arctan((a+bk)^{-1}) \geq (a+bk)^{-1} \frac{\arctan((a+b)^{-1})}{(a+b)^{-1}} \geq \arctan((a+b)^{-1}) k^{-1},
\end{equation*}
and hence
\begin{equation*}
  \psi = 2 \arctan \Bigl( \frac{R}{2(1+(k-1/2)\sqrt{3})R} \Bigr) \geq \frac{1}{2k}
\end{equation*}
for any $k \geq 1$. Thus $B_k$ can be covered with $\lceil 4\pi k \rceil$ closed $R$-balls, and therefore cannot contain more than $\lceil 4\pi k \rceil T(R)$ points of~$\eta$. 
\end{proof}

\subsection{Random transmitter positions}
\label{ssec:randomtransmitter}

Say now the transmitter positions form a point process $\Xi = \{X_i \colon i \in \mcI^{\Xi} \} = \sum_{i\in\mcI^{\Xi}} \delta_{X_i}$
(by the identification at the beginning of Section~\ref{ssec:model}) as described in Setup~\ref{setup}. Suppose that in addition to the mean measure $\lambda$ the second factorial moment measure $\lambda_{[2]}$ exists, meaning that
\ba{
  \lambda(A) &= \mean \Xi(A) < \infty \quad \text{for every bounded Borel set $A \subset \IR^2$}\\
  \lambda_{[2]}(B) &= \mean \Xi^{[2]}(B) < \infty \quad \text{for every bounded Borel set $B \subset \IR^4$},
}
where $\Xi^{[2]} = \sum_{i \neq j} \delta_{(X_i,X_j)}$.

For a Poisson process with mean measure $\lambda$ it is easily checked that $\lambda_{[2]} = \lambda \otimes \lambda$; see Example~9.5(d) in \cite[Section~9.5]{Daley2008}. For a hard core process with minimal distance $\eps_*$, the maximal number of points that can lie in a fixed bounded set is bounded, so the second factorial moment measure always exists; it is readily checked that
\ben{
\lambda_{[2]}(\{(x,y) \in \IR^2\times \IR^2 \colon \norm{y-x} < \eps_*\}) = 0.  
}

A direct consequence of the definition of these moment measures is that for general~$\Xi$ and non-negative measurable functions $h_1$, $h_2$,
\ban{
\mean \int_{\IR^2} h_1(x)  \Xi (dx)&=\int_{\IR^2} h_1(x) \; \lambda(dx), \label{eq:camp1}\\
\mean \int_{\IR^2\times \IR^2} h_2(x, y) \Xi(dx) \Xi(dy)&=\int_{\IR^2\times \IR^2}  h_2(x,y) \; \lambda_{[2]}(dx \times dy) +\int_{\IR^2} h_2(x,x) \; \lambda(dx);\label{eq:camp2}
}
see for example \cite[Section~9.5]{Daley2008}. The first expression~\eq{eq:camp1} is sometimes
called Campbell's theorem.

Denote now by $N=N\t{\Xi,\sigma}$ and $\Pi$ 
the corresponding processes based on $\Xi$ instead of the fixed $\xi$, 
where the signal strengths $\{S_x, x \in \R^2\}$ are assumed to be independent of~$\Xi$. 
Note that $N(t)=\int_{\IR^2}  \I[g(x)/S_x \leq t] \, \Xi(dx)$ and
by independence and~\eq{eq:camp1},
\ben{
  M(t) = \mean N(t) = \mean \int_{\IR^2} \pr(g(x)/S_x \leq t \mvert \Xi) \; \Xi(dx) = \int_{\IR^2} \pr(g(x)/S_x \leq t) \; \lambda(dx).  \label{eq:meanofrandom}
}
Theorem~\ref{thm:meanlim} implies that $M(t) = L(t) = \kappa \pi t^{2/\beta} / K^2$, which does not depend on $\sigma$.

\begin{Theorem}\label{thm:randnumpd2}
Let $t>0$ and $\sigma,\Xi$ and $N := N^{(\Xi,\sigma)}$ be defined as in Setup~\ref{setup}. Use $\Pi$ to denote the Poisson process on $\Rplus$ that has the same mean measure $M$ as $N$. Denote the conditional mean measure of $N$ given $\Xi$ by $M^{\Xi}$.
Let $R>0$ be large enough to apply {\rm P2}.
Fix positive constants $C\geq R$, $d \leq C$ 
such that
$b := b(d) =\frac{1}{\sigma} \log ( h(d)/t )+ \frac{\sigma}{\beta}>0$.
Also set
$B_C=\frac{1}{\sigma} \log ( h(C)/t )+ \frac{\sigma}{\beta}$
and $\delta(\eps)$ to be the  uniform positive definite
constant of {\rm P1} for $\eps>0$. 
Furthermore for any $T^*>0$, 
set
\be{
F_{T^*,\eps}:=F(R,T^*,\eps)=\frac{1}{\delta(\eps)} (4\pi+1) T^* \biggl( \tilde\varrho^2(R) + \frac{1}{\sqrt{3}R^2} \int_R^\infty s \tilde\varrho^2(s) \; ds \biggr),
}
and assume that $B_C > 1$, and $B_C^2 F \leq 1$.

\begin{enumerate}
\item[(i)] Let $\Xi$ be a homogeneous Poisson process with intensity~$\kappa$, assume that $T^*\geq 16\kappa R^2$, choose further constants
$\eps_0,\eps_C > 0$, and set $F := F_{T^*,\eps_C}$. Then
\ba{
\dtwo(&\law(N\vert_t),\law(\Pi\vert_t)\\
&\leq 2 \kappa \int_{\norm{x}>C}\pr(t S\geq g(x)) \; dx+\mean\left|M^\Xi(t)-M(t)\right| 
+\E \int_0^t \bigl|M^\Xi(s)-M(s)\bigr| \; ds \\
&\qquad 
 + M(t) \biggl[ (\kappa \pi R^2+1) \pr(t S\geq h(d)) + 5\kappa \pi  \Bigl( \eps_0^2 + R^2 e^{-b^2(1-\tilde\varrho(\eps_0))/4} \Bigr) \biggr] \\
&\qquad {} + (t+1)M(t)\left[\frac{8 (B_C+\sigma^{-1}) \sqrt{F}}{\sqrt{1-F^2}}+
	(1+b^{-2}) \sqrt{F} e^{-b^2(F^{-1}-1)/2} \right] \\
 &\qquad
 +\kappa\pi d^2
+ (C/R+1)^2 \exp\left\{-T^*\left(\log\left(\frac{T^*}{16\kappa R^2}\right)-1\right)-16\kappa R^2\right\}\\
&\qquad 
  +(2 C/\eps_C+1)^2 (4\kappa  \eps_C^2)^2.
}

\item[(ii)] Let $\Xi$ be an intensity $\kappa$ second order stationary hard core process with distance $\eps_{*}>0$, fix $T^* :=  4 \bigl(\frac{R+\eps_*/2}{\eps_*} \bigr)^2$, and set $F := F_{T^*,\eps_*}$.
Then, 
\ba{
\dtwo(&\law(N\vert_t),\law(\Pi\vert_t)\\
&\leq 2 \kappa \int_{\norm{x}>C}\pr(t S\geq g(x)) \; dx+\mean\left|M^\Xi(t)-M(t)\right| 
+\E \int_0^t \bigl|M^\Xi(s)-M(s)\bigr| \; ds \\
&\qquad 
 + M(t) T^* \biggl[  \pr(tS \geq h(d))+ 5   e^{-b^2(1-\tilde\varrho(\eps_*))/4} \biggr] +\kappa\pi d^2 \\
&\qquad {} + (t+1)M(t)\left[\frac{8 (B_C+\sigma^{-1}) \sqrt{F}}{\sqrt{1-F^2}}+
	(1+b^{-2}) \sqrt{F} e^{-b^2(F^{-1}-1)/2} \right].
}
\end{enumerate}
\end{Theorem}

We defer the discussion about convergence as $\sigma \to \infty$ until after the proof; see Theorem~\ref{thm:randprocconv} below.
\begin{proof}
For the initial estimates the nature of the process $\Xi$ (Poisson or hard core) is not important. 

Note that~\eq{eq:dtwoassup} implies $\dtwo$ can be written as a supremum of absolute differences of expectations over a class of functions $\mcF$, so that
for any point processes $A$ and $B$,
\ba{
  \dtwo(\law(A), \law(B)) &= \sup_{f \in \mcf} \, \bigl| \mean f(A) - \mean f(B) \bigr| \\
  &= \sup_{f \in \mcf} \, \bigl| \mean \bigl( \mean (f(A) \mvert \Xi) - \mean (f(B) \mvert \Xi) \bigr) \bigr| \\
  &\leq  \mean \sup_{f \in \mcf} \, \bigl| \mean (f(A) \mvert \Xi) - \mean (f(B) \mvert \Xi) \bigr| \\
  &= \mean \dtwo(\law(A\mvert \Xi), \law(B\mvert \Xi)).
}
 
Denote by $\Pi^{\Xi}$ the point process which conditionally on $\Xi$ is a Poisson process with (conditional) mean measure $M^{\Xi}$. Thus, $\Pi^{\Xi}$ is a so-called Cox process, a mixture of Poisson processes. 

Split up the original distance as 
\ban{
  \dtwo(\law(N\vert_t), \law(\Pi\vert_t)) \leq \dtwo(\law(N\vert_t), \law(\Pi^{\Xi}\vert_t)) + \dtwo(\law(\Pi^{\Xi}\vert_t), \law(\Pi \vert_t))
}

We upper bound the second term by conditioning on $\Xi$ and 
using Lemma~\ref{lem:ppp} in Appendix~\ref{sec:proofppp} below, which bounds the $\dtwo$ distance between two Poisson point processes
on $[0,t]$. Thus
\ban{
\dtwo(\law(\Pi^\Xi\vert_t), \law(\Pi\vert_t)) &\leq \mean \dtwo(\law(\Pi^\Xi\vert_t \mvert \Xi), \law(\Pi\vert_t)) \notag\\
&\leq \mean\int_0^t \bigl|M^\Xi(s)-M(s)\bigr| \; ds + \mean \bigl|M^\Xi(t)-M(t)\bigr|. \label{eq:poissoncoxerror}
}

Conditioning the first term also on $\Xi$ yields
\be{
 \dtwo(\law(N\vert_t), \law(\Pi^{\Xi}\vert_t)) \leq \mean \dtwo(\law(N\vert_t \mvert \Xi), \law(\Pi^{\Xi}\vert_t \mvert \Xi)).
}
The term inside the mean corresponds precisely to the term bounded in Theorem~\ref{thm:mainproc}, except that not every transmitter configuration $\Xi$ satisfies all of the minimal distance and maximal point count requirements imposed for the deterministic configuration $\xi$ in Theorem~\ref{thm:mainproc}. Following the proof of that theorem we truncate the configurations to $\bar{B}(0,C)$, which did not use any conditions. 
Set $\Xi_C=\Xi\vert_{\bar{B}(0,C)}$, $N_C=N^{(\Xi_C,\sigma)}$,
and $\Pi_C^\Xi=\Pi^{\Xi_C}$.
Since
\ben{
   \mean \int_{\norm{x}>C} \pr(t S \geq g(x)) \; \Xi(dx) = \int_{\norm{x}>C}\pr(t S\geq g(x)) \; \lambda(dx)
   \label{eq:cutofferror}
}
in the same way as \eqref{eq:meanofrandom}, we obtain
\ban{\mean \dtwo(\law(N\vert_t \mvert \Xi), \law(\Pi^{\Xi}\vert_t \mvert \Xi)) &\leq
2 \int_{\norm{x}>C}\pr(t S\geq g(x)) \; \lambda(dx) \notag\\
&\hspace*{20mm} {} + \mean \dtwo \bigl( \law(N_C\vert_t \mvert \Xi), \law(\Pi_C^\Xi\vert_t \mvert \Xi) \bigr).
}
We may now apply Theorem~\ref{thmppp} to the term inside the mean
on the good event that the transmitters are not too close to each other or the origin
and then bound the probability that the good event does not occur.

To this end, let
$D_{\Xi_C}=\min_{x\in\Xi_C} \norm{x}$,
$E_{\Xi_C}= \inf_{x,y \in \Xi_C, x\neq y} \norm{x-y}$, 
and
$T_{\Xi_C}(R)=\max_{x\in\Xi_C} \Xi_C(\bar B(x,R))$.
For the positive constants $d, T^*$, and setting $\eps = \eps_C$ in the Poisson case and $\eps = \eps_*$ in the hard core case, define the events 
\be{
\mathcal{D}=\{D_{\Xi_C}\geq d\}, \hspace{0.4cm}\mathcal{E}=\{E_{\Xi_C}\geq \eps\},\hspace{0.4cm} \mathcal{T}=\{T_{\Xi_C}(R) \leq T^*\}.
}
We first bound
\ba{
\mean \dtwo&\bigl( \law(N_C\vert_t \mvert \Xi), \law(\Pi_C^\Xi\vert_t \mvert \Xi) \bigr) \\
&\leq \mean \Bigl( \dtwo\bigl( \law(N_C\vert_t \mvert \Xi), \law(\Pi_C^\Xi\vert_t \mvert \Xi) \bigr) \hbit \I[\mathcal{D}\cap\mathcal{T}\cap\mathcal{E}] \Bigr) +\pr(\mathcal{D}^c)+\pr(\mathcal{E}^c)+\pr(\mathcal{T}^c).
}

Note that
\ben{
  \pr(\mathcal{D}^c)= \pr(\Xi({\mathring B}(0,d)) > 0) \leq \mean(\Xi({\mathring B}(0,d))) = \kappa \pi d^2; \label{eq:dcomp}   
}
the remaining two probabilities are bounded further below, because we have to distinguish whether $\Xi$ is Poisson or hard core.

For the $\dtwo$ term, conditional on $\Xi$ 
and under the good event, 
we apply Theorem~\ref{thmppp} with $A_i=\{j \in \mcI^{\Xi}:\norm{X_i-X_j}\leq R\}$
 and bound the four terms in~\eqref{eq:pppdtwot1} and \eqref{eq:pppdtwo} in the analogous way as in Theorem~\ref{thm:mainproc} taking always the $1$ in the minimum in the prefactor. 

For the first and second terms of~\eq{eq:pppdtwot1} we have to distinguish whether $\Xi$ is Poisson or hard core, so we postpone this part to further below. For the third term of~\eq{eq:pppdtwot1} and
the term~\eq{eq:pppdtwo}, we obtain immediately from~\eqref{eq:totbound3rd} and \eqref{eq:intwdpout} (noting that the $M(t) = M\t{\xi,\sigma}(t)$ in those formulae are integrals with respect to $\xi$)
\ban{
  (t+1) \int_{\IR^2} \pr(Z > b_x) \; \Xi(dx) \hbit \biggl[ \frac{8 (B_C + \sigma^{-1}) \sqrt{F} }{\sqrt{1-F^2}}
  + (1+b^{-2})\sqrt{F} e^{-b^2(F^{-1}-1)/2} \biggr], \label{eq:lasttermerror_noexp}
}
and after taking expectations this yields the required bound.

It remains to bound the first and second term of~\eq{eq:pppdtwot1}
 and also $\pr(\mathcal{E}^c)$ and $\pr(\mathcal{T}^c)$.
\bigskip

\noindent\textit{(i) \hbit The Poisson process case.}
\smallskip

\noindent
For the first term in~\eq{eq:pppdtwot1}, we obtain
\ba{
 \I[\mathcal{D}\cap\mathcal{E}\cap\mathcal{T}]&\int_{\IR^2} \int_{\bar{B}(x,R)} \pr(t S\geq g(y))\pr(t S\geq g(x)) \; \Xi_C(dy) \, \Xi_C(dx) \notag \\
&\qquad \qquad \leq  \pr(tS \geq h(d)) \int_{\R^2} \int_{\bar{B}(x,R)}\pr(tS \geq g(x)) \; \Xi(dy) \, \Xi(dx).
}
After taking expectation, using~\eq{eq:camp2}, this yields an upper bound of
\ban{
\pr(tS \geq h(d)) &\biggl(\int_{\R^2} \int_{\bar{B}(x,R)}  \pr(tS \geq g(x)) \kappa^2 \; dy \, dx + \int_{\IR^2} \pr(tS \geq g(x)) \kappa \; dx  \biggr) \notag  \\[1mm]
&\hspace*{-7mm}=(\kappa \pi R^2+1) M(t) \pr(tS \geq h(d)).  \label{eq:firsttermerrorppp}
}

For the second term in \eqref{eq:pppdtwot1}, we first note that the arguments in the proof of Theorem~\ref{thm:mainproc} from \eqref{35} to \eqref{337} imply that, for $x,y \in \R^2$ with $\norm{y} \geq \norm{x} \geq d$ (so that $b_y \geq b_x \geq b > 0$),
\ben{
 \pr(t S_{x}\geq g(x), t S_{y}\geq g(y))\leq \frac{5}{2} \exp\left\{-\frac{b^2(1-\rho(x,y))}{4}\right\} \pr(t S\geq g(x)), \label{eq:jointbd}
}
where $b = b(d) = \frac{1}{\sigma} \log \bigl( \tfrac{h(d)}{t} \bigr) + \frac{\sigma}{\beta}$. Thus, similarly to the argument starting from~\eqref{333},
\ba{
 &\I[\mathcal{D}\cap\mathcal{E}\cap\mathcal{T}] \int_{\IR^2} \int_{\bar B(x,R) \setminus \{x\}} \pr(t S_{x}\geq g(x), t S_{y}\geq g(y)) \; \Xi_C(dy) \, \Xi_C(dx) \notag \\
&\leq 2 \hbit \I[\mathcal{D}] \int_{\IR^2}\mathop{\int_{\bar B(x,R) \setminus \{x\}}}_{\norm{y}\geq \norm{x}} \pr(t S_{x}\geq g(x), t S_{y}\geq g(y)) \; \Xi(dy) \, \Xi(dx) \notag \\
&\leq 5 \int_{\IR^2} \int_{\bar B(x,R) \setminus \{x\}} \exp\left\{-\frac{b^2(1-\rho(x,y))}{4}\right\} \pr(t S\geq g(x)) \; \Xi(dy) \, \Xi(dx) \notag \\
&\leq 5 \biggl( \int_{\IR^2} \int_{\bar B(x,R) \setminus \bar B(x,\eps_0)} \exp\left\{-\frac{b^2(1-\tilde \varrho(\eps_0))}{4}\right\} \pr(t S\geq g(x)) \; \Xi(dy) \, \Xi(dx) \notag \\
&\qquad\qquad{} + \int_{\IR^2} \int_{\bar B(x,\eps_0)\setminus \{x\}} \pr(t S\geq g(x)) \; \Xi(dy) \, \Xi(dx) \biggr).
}
After taking expectations using~\eq{eq:camp2}, this yields as an upper bound of the second term
\ban{
5 \biggl( &\int_{\IR^2} \int_{\bar B(x,R) \setminus \bar B(x,\eps_0)} \exp\left\{-\frac{b^2(1-\tilde \varrho(\eps_0))}{4}\right\} \pr(t S\geq g(x)) \kappa^2 \; dy \, dx \notag \\
&\qquad {} + \int_{\IR^2} \int_{\bar B(x,\eps_0)} \pr(t S\geq g(x)) \kappa^2 \; dy \, dx \biggr) \notag \\
&\leq 5 \kappa \pi R^2 M(t) e^{-b^2(1-\tilde\varrho(\eps_0))/4} + 5 \kappa \pi \eps_0^2 M(t). \label{eq:secondtermerror}
}

For the remaining terms, note that
\ba{
\pr(\mathcal{E}^c)&=\pr(E_{\Xi_C} < \eps_C)=\pr(T_{\Xi_C}(\eps_C/2)\geq 2), \\ 
\pr(\mathcal{T}^c)&=  \pr(T_{\Xi_C}(R)> T^*). 
}
In general, if $T_{\Xi_C}(r)\geq M$ then at least one of the set of  $(2 \lceil C/(2r) \rceil-1)^2$ squares 
with side length $4r$ and centers
at the points $\{(2ri, 2rj): i,j=-\lceil C/(2r) \rceil + 1, \ldots, \lceil C/(2r) \rceil - 1 \}$ 
must have at least $M$ points in it. This is because for any disc $D$ of radius $r$ the set $D \cap {\bar B(0,C)}$ is completely contained in one of the squares in the set (since it intersects at most four of the $2r \times 2r$ squares formed by overlap which must have a common corner). 
Thus if $\square$ denotes the $4r\times 4r$ square with center at the origin, then
\be{
\pr(T_{\Xi_C}(r)\geq M)\leq (C/r+1)^2 \pr(\Xi_C(\square)\geq  M).
} 
For $\mathcal{E}^c$, we then have
\ban{
\pr(T_{\Xi_C}(\eps_C/2)\geq 2)&\leq  (2 C/\eps_C+1)^2 (1-\exp\{-4 \kappa \eps_C^2\}(1+4 \kappa \eps_C^2))  \notag\\
	&\leq (2 C/\eps_C+1)^2 (4\kappa  \eps_C^2)^2. \label{eq:ecomp}
	}
For $\mathcal{T}^c$, let $X_\mu$ denote a Poisson variable with mean $\mu$. Then
since $T^*\geq 16\kappa R^2$,
\ban{
 \pr(T_{\Xi_C}(R)> T^*)&\leq  (C/R+1)^2 \pr(X_{16 \kappa R^2} \geq T^*) \notag\\
 	&\leq (C/R+1)^2 \exp\left\{-T^*\left(\log\left(\frac{T^*}{16\kappa R^2}\right)-1\right)-16\kappa R^2\right\}, \label{eq:tcomp} 
}
where we have used the Poisson Chernoff bound, for $s\geq \mu$,
\be{
\pr(X_\mu\geq s)\leq \exp\{-s(\log(s/\mu)-1)\}-\mu\}. 
}

Collecting \eqref{eq:poissoncoxerror}, \eqref{eq:cutofferror}, \eqref{eq:firsttermerrorppp}, \eqref{eq:secondtermerror}, the expectation of \eqref{eq:lasttermerror_noexp}, as well as \eqref{eq:dcomp}, \eqref{eq:ecomp}, and \eqref{eq:tcomp} yields the required upper bound.
\bigskip

\noindent\textit{(ii) \hbit The hard core process case.}
\smallskip

\noindent
Note first that the maximum number of points of $\Xi$ contained in any ball of radius $R$ can be bounded by the maximum number of non-overlapping balls of radius $\eps_*/2$ that can be placed into a ball of radius $R+\eps_*/2$. By comparing the areas of the balls we obtain
\ben{
  T_{\Xi}(R)=\max_{x\in\Xi} \Xi(\bar B(x,R)) \leq 4 \biggl(\frac{R+\eps_*/2}{\eps_*} \biggr)^2 = T^*,
  \label{eq:maxnopoints}
}

For the first term in~\eq{eq:pppdtwot1}, we argue similarly to
the Poisson process case to find, after taking expectation, an upper bound of 
\ban{
\mean \int_{\R^2}& \int_{\bar{B}(x,R)} \pr(tS \geq h(d)) \pr(tS \geq g(x)) \; \Xi(dy) \, \Xi(dx) \notag \\
&= \pr(tS \geq h(d)) \, \mean \int_{\R^2} \Xi(\bar{B}(x,R)) \pr(tS \geq g(x)) \; \Xi(dx)\notag\\
&\leq \pr(tS \geq h(d)) \, T^* \hbit \mean \int_{\R^2} \pr(tS \geq g(x)) \; \Xi(dx)\notag\\
&=M(t) \hbit T^* \hbit \pr(tS \geq h(d)).   \label{eq:firsttermerror}
}

For the second term in \eqref{eq:pppdtwot1}, we obtain in the analogous way as for the Poisson process case 
\ban{
 &\I[\mathcal{D}\cap\mathcal{E}\cap\mathcal{T}] \int_{\IR^2} \int_{\bar B(x,R) \setminus \{x\}} \pr(t S_{x}\geq g(x), t S_{y}\geq g(y)) \; \Xi_C(dy) \, \Xi_C(dx) \notag \\
&\leq 2 \hbit \I[\mathcal{D}] \int_{\IR^2}\mathop{\int_{\bar B(x,R) \setminus \{x\}}}_{\norm{y}\geq \norm{x}} \pr(t S_{x}\geq g(x), t S_{y}\geq g(y)) \; \Xi(dy) \, \Xi(dx) \notag \\
&\leq 5 \int_{\IR^2} \int_{\bar B(x,R) \setminus \{x\}} \exp\left\{-\frac{b^2(1-\tilde{\varrho}(\eps_*))}{4}\right\} \pr(t S\geq g(x)) \; \Xi(dy) \, \Xi(dx) \notag\\[1mm]
&\leq 5 M(t) \hbit T^* e^{-b^2(1-\tilde{\varrho}(\eps_*))/4}, \label{eq:secondtermerrorhc}
}
where we used \eqref{eq:maxnopoints} again for the last inequality.

For the remaining terms, note that
\ba{
\pr(\mathcal{E}^c)&\leq \pr(E_{\Xi} < \eps_*) = 0, \\ 
\pr(\mathcal{T}^c)&\leq \pr(T_{\Xi}(R) > T^*) = 0. 
}

Collecting \eqref{eq:poissoncoxerror}, \eqref{eq:cutofferror}, \eqref{eq:firsttermerror}, \eqref{eq:secondtermerrorhc}, the expectation of \eqref{eq:lasttermerror_noexp}, as well as \eqref{eq:dcomp} yields the required upper bound. \qedhere
\end{proof}

In order to use Theorem~\ref{thm:randnumpd2} to show a convergence result for the hard core process, we
assume a mild ``$\mathbf{B}^+_2$-mixing'' condition which we now describe.

For a second-order stationary point process with
intensity~$\kappa$, \cite[Proposition~12.6.III]{Daley2008} yields that there exists a ``reduced'' measure $\breve{\lambda}_{[2]}$ on $\IR^2$ satisfying
\ben{
  \int_{\IR^2 \times \IR^2} h(x,y) \; \lambda_{[2]}(dx \times dy) = \int_{\IR^2} \int_{\IR^2} h(x,x+y) \hbit \kappa \; dx \; \breve{\lambda}_{[2]}(dy)
  \label{eq:reduction}
}
for every measurable function $h \colon \IR^2 \times \IR^2 \to \Rplus$ (note that we normalize differently than \cite{Daley2008}). As can be seen from \cite[Proposition~13.2.VI]{Daley2008} and the discussion following it, we may interpret $\breve{\lambda}_{[2]}(B)$ as the expected number of \emph{other} transmitters in the set $x+B$ given there is a transmitter at $x$.

We assume that there exists a (necessarily locally finite) signed measure $\breve{\gamma}_{[2]}$ on $\IR^2$, called the \emph{reduced covariance measure}, which satisfies
\ben{
  \breve{\gamma}_{[2]}(B) = \breve{\lambda}_{[2]}(B) - \lambda(B)  \label{eq:defbgamma}
}
for any bounded Borel set $B \subset \IR^2$. Furthermore we assume that $\breve{\gamma}_{[2]}$ is $[-\infty,\infty)$-valued, which is the same as saying that the \emph{positive variation} $\breve{\gamma}^+_{[2]}$ (the positive part in the Jordan decomposition of $\breve{\gamma}_{[2]}$) is \emph{finite}. 
We refer to the existence of a reduced covariance measure of finite positive variation as \emph{$\mathbf{B}^+_2$-mixing}, with regard to the frequently used stronger concept of $\mathbf{B}_2$-mixing, which means that $\breve{\gamma}_{[2]}$ (or equivalently its variation $|\breve{\gamma}_{[2]}|$)
is finite; see e.g.\ \cite{Heinrich2013}.

Second order stationary hard core point processes with $\mathbf{B}_2$- (hence also $\mathbf{B}_2^+$-) mixing form a rich family which covers most reasonable models of transmitter configurations with
minimal distance.
If we thin a $\mathbf{B}_2$-mixing process according to a stationary (and typically dependent) thinning procedure given by a random field $\{\pi_x, x \in \IR_2\}$ of retention probabilites satisfying $\int_{\IR^2} \cov(\pi_0,\pi_x) \, dx < \infty$, then the thinned process is still $\mathbf{B}_2$-mixing; see \cite[Example~4, Section~6]{Heinrich2013}. Likewise, if we start from a $\mathbf{B}_2$-mixing process of cluster centers and replace each center $c$ by a random number of i.i.d.\ points distributed according to a fixed distribution that is shifted to $c$, then the result is $\mathbf{B}_2$-mixing again if the distribution of the cluster sizes has second moment; see \cite[Example~4.1]{Heinrich2013a}. The possibility of using either or both of these operations on any of the known stationary $\mathbf{B}_2$-mixing processes, such as the homogeneous Poisson process, any stationary log-Gaussian Cox process with integrable covariance function of the underlying random field (seen from \cite[Proposition~5.4]{Moller2004}), or any stationary determinantal point process \cite{Biscio2016}, gives us a wide choice of possible models.
As long as the last thinning step leaves the model with a hard core, item (ii) of the theorem below is applicable. Arguably the easiest natural way to introduce a hard core is by the Mat\'ern ``type II'' thinning procedure described in \cite[Section~5.4]{Stoyan1987}: The points are equipped with i.i.d.\ marks that can be interpreted as ``virtual birth (or construction) times'' of the points; the thinned process consists then of those points of the original process that are older than any other points within a pre-specified hard core distance $\eps_*>0$. 

To understand both $\mathbf{B}_2$- and $\mathbf{B}^{+}_2$-mixing intuitively,
note that for disjoint bounded Borel sets $A,B \subset \IR^2$, we have by \eqref{eq:reduction} and the translation invariance of Lebesgue measure
\be{
  \cov(\Xi(A),\Xi(B)) = \lambda_{[2]}(A \times B) - \kappa^2 |A| \hbit |B| = \int_{\IR^2} \int_{\IR^2} \I[x \in A] \hbit \I[x+y \in B] \hbit \kappa \; dx \; \breve{\gamma}_{[2]}(dy).
}
Thus, for any bounded Borel set $A \subset \IR^2$ with $\diam(A) := \sup_{x,y \in A} \norm{x-y}$ and any (not necessarily bounded) Borel sets $B_n \subset {\bar B(0,n)}^c$, $n \in \NN$, we have as $n \to \infty$
\be{
  \bigl| \cov(\Xi(A),\Xi(B_n)) \bigr| \leq \kappa \cdot |A| \cdot |\breve{\gamma}_{[2]}|({\bar B(0,n-\diam(A))}^c) \to 0
}
under $\mathbf{B}_2$-mixing since $\tvnorm{\breve{\gamma}_{[2]}} = |\breve{\gamma}_{[2]}|(\IR^2) < \infty$, and
\be{
  \bigl(\cov(\Xi(A),\Xi(B_n))\bigr)^{+} \leq \kappa \cdot |A| \cdot \breve{\gamma}^{+}_{[2]}({\bar B(0,n-\diam(A))}^c) \to 0
}
under $\mathbf{B}^{+}_2$-mixing since $\breve{\gamma}^{+}_{[2]}(\IR^2) < \infty$.

We state our convergence result.
\begin{Theorem}\label{thm:randprocconv}
Let $t>0$ and $\sigma,\Xi$ and $N := N^{\Xi,\sigma}$ be defined as in Setup~\ref{setup}. Use $\Pi$ to denote the Poisson process on $\Rplus$ that has the same mean measure $M$ as $N$. Denote the conditional
mean measure of $N$ given $\Xi$ by $M^{\Xi}$.
If either
\begin{enumerate}[label=(\roman*)]
\item $\Xi$ is a homogeneous Poisson process with intensity $\kappa>0$, and for some $c>0$, 
we may take for the uniform positive definite
constant of {\rm P1}: $\delta=\delta(\eps)=\Omega(\eps^c)$ as $\eps\to0$, and for some $a>0$,
 we have that  $\tilde\varrho(r)=\textrm{O}(e^{-ar})$ as $r\to\infty$, or
 \item $\Xi$ is a $\mathbf{B}^{+}_2$-mixing second order stationary hard core process with distance $\eps_{*}>0$ and intensity $\kappa>0$,
 and 
$\tilde\varrho(r)=\textrm{O}(r^{-(1+a)})$ for some $a>0$,
 \end{enumerate}
  then as $\sigma\to\infty$,
 \be{
 \dtwo(\law(N\vert_t), \law(\Pi\vert_t)) \to 0.
 }
\end{Theorem}
\begin{proof}
We apply Theorem~\ref{thm:randnumpd2}. 
For both cases, set $C=\exp\{\sigma^2/\beta^2+\sigma^{1.1}\}$
so that, following the proof of Theorem~\ref{thm:fixprocconv}, the first term in the bound goes to zero. For $\sigma$ large enough, we can set $d =\sigma^{-2}$ and the condition that $b = b(d) > 0$ is satisfied. Note that $b = b(d) = \Theta(\sigma)$. Using that for any $\gamma>0$,
\ben{\label{eq:markbd}
\pr(tS\geq h(\gamma))\leq \frac{t^{1/\beta}\mean S^{1/\beta}}{K \gamma}=\frac{t^{1/\beta}e^{-\sigma^2/(2\beta^2)}}{K \gamma},
}
we find that for any $p>0$, if $R=\sigma^p$, then
\be{
R^2\pr(t S\geq h(d))\leq   \sigma^{2+2p} \frac{t^{1/\beta}e^{-\sigma^2/(2\beta^2)}}{K}\to 0.
}

For the Poisson case~$(i)$, set $\eps_C=\exp\{-\sigma^2/\beta^2-2\sigma^{1.1}\}$ so that $C\eps_C\to 0$
and set $T^*=16 e \kappa R^2$. Then note that the conditions on $\tilde\varrho$
imply $F=\textrm{O}(\eps_C^{-c} R^2 e^{-2 a R})$
 and set $R=\sigma^3$
so that $B_C\sqrt{F}\to 0$ (since $B_C=\textrm{O}(\sigma)$).
We further set $\eps_0:=\sigma^{-1}+\sup\{\eps\geq 0: 1-\tilde\varrho(\eps)\leq \sigma^{-1}\}$ which
 is well-defined and tends to zero as $\sigma\to\infty$ since $\tilde\varrho$
 is non-increasing to zero with $\tilde\varrho(r)=1$ if and only if $r=0$.

For the hard core case~$(ii)$, the tail conditions on $\tilde \varrho$ imply
that $F=\textrm{O}(R^{-2a})$, so choosing $R=\sigma^{2/a}$ ensures
$B_C\sqrt{F}\to 0$.

With these choices of parameters, it is straightforward to see that for both cases,
almost all of the terms 
from Theorem~\ref{thm:randnumpd2} tend to zero as $\sigma\to\infty$; the exceptions are $\mean\left|M^\Xi(t)-M(t)\right|$
and 
\ben{\label{eq:intbdmean}
\mean \int_0^t\bigl|M^\Xi(s)-M(s)\bigr| \; ds.
}

Note that by Cauchy--Schwarz $\mean| M^{\Xi}(s)-M(s)|\leq\sqrt{\var ( M^{\Xi}(s))}$ for $0\leq s \leq t$. We show $\var ( M^{\Xi}(s)) \to 0$ uniformly in $s \in [0,t]$ (and note the calculations below justify the finiteness of the second moment of $M^{\Xi}(s)$), which by applying  Fubini shows that~\eqref{eq:intbdmean} tends to zero.

For both the Poisson and the hard core case, use~\eq{eq:camp2} to find
\ban{
\var ( M^{\Xi}(s))&=\mean M^{\Xi}(s)^2 - M(s)^2 \notag\\
  &=\mean \int_{\IR^2\times \IR^2} \pr(sS\geq g(x)) \pr(sS\geq g(y)) \; \Xi(dx) \, \Xi(dy)-M(s)^2\notag\\
  \begin{split}\label{eq:varM}
  &=\int_{\IR^2} \pr(sS\geq g(x))^2 \lambda(dx)\\
  &\quad {} + \int_{\IR^2\times \IR^2} \pr(sS\geq g(x)) \pr(sS\geq g(y)) \; \lambda_{[2]}(dx \times dy)\\
  &\quad {} - \int_{\IR^2\times \IR^2} \pr(sS\geq g(x)) \pr(sS\geq g(y)) \; \lambda(dx) \, \lambda(dy). 
  \end{split}
}

In the Poisson case~$(i)$, we have $\lambda_{[2]} = \lambda \otimes \lambda$, so that
\ba{
\var ( M^{\Xi}(s))&=\int_{\IR^2} \pr(sS\geq g(x))^2 \; \lambda(dx) \\
	&\leq \int_{\norm{x}\leq \sigma^{-1}} \pr(sS\geq g(x))^2 \; \lambda(dx)
		+ \int_{\norm{x}> \sigma^{-1}} \pr(sS\geq g(x))^2 \; \lambda(dx)\\
	&\leq \kappa \pi \sigma^{-2}+ M(s) \hbit \pr(sS\geq h(\sigma^{-1}))\\
	&\leq \kappa \pi \sigma^{-2}+ M(s) \hbit K^{-1}s^{1/\beta}\sigma e^{-\sigma^2/(2\beta^2)} \\
	&\leq \kappa \pi \sigma^{-2}+ M(t) \hbit K^{-1}t^{1/\beta}\sigma e^{-\sigma^2/(2\beta^2)} 
	 \stackrel{\sigma\to\infty}{\longrightarrow} 0,
}
where in the second last inequality we use~\eq{eq:markbd}.

In the hard core case~$(ii)$, we note that the first term in \eqref{eq:varM} is the same as in the Poisson case and thus goes to zero uniformly in $s \in [0,t]$. By $\mathbf{B}^{+}_2$-mixing the positive variation $\breve{\gamma}^{+}_{[2]}$ of $\breve{\gamma}_{[2]}$ is a finite measure. Assume $\sigma > 2/\eps_*$ and note that $D_\sigma := \{(x,y) \in \IR^2 \times \IR^2 \colon \norm{x} \leq 1/\sigma, \norm{y} \leq 1/\sigma \} \subset \{(x,y) \in \IR^2 \times \IR^2 \colon \norm{y-x} < \eps_* \}$
on which $\lambda_{[2]}$ is zero. Write furthermore $\tD_{\sigma} := \{(x,y) \in \IR^2 \times \IR^2 \colon \norm{x} \leq 1/\sigma, \norm{x+y} \leq 1/\sigma \}$.

Using \eqref{eq:reduction} and the fact that $\lambda(dx) = \kappa \, dx$ is translation invariant, we obtain as an upper bound of the remaining terms of \eqref{eq:varM}, showing for the equality below first that it holds when integrating over $\tD_{\sigma}^c \cap \bigl(\IR^2 \times {\bar B}(0,n)\bigr)$ for arbitrary $n \in \NN$,
\ba{
  &\int_{\tD_{\sigma}^c} \pr(sS\geq g(x)) \pr(sS\geq g(x+y)) \; \lambda(dx) \, \breve{\lambda}_{[2]}(dy) \\
  &\hspace*{25mm} {} - \int_{\tD_{\sigma}^c} \pr(sS\geq g(x)) \pr(sS\geq g(x+y)) \; \lambda(dx) \, \lambda(dy) \\
  &\hspace*{10mm} = \int_{\tD_{\sigma}^c} \pr(sS\geq g(x)) \pr(sS\geq g(x+y)) \; \lambda(dx) \, \breve{\gamma}_{[2]}(dy) \\
  &\hspace*{10mm} \leq \int_{(x,y) \colon \norm{x} > 1/\sigma} \pr(sS\geq g(x)) \pr(sS\geq g(x+y)) \; \lambda(dx) \, \breve{\gamma}^{+}_{[2]}(dy) \\
  &\hspace*{25mm} {} + \int_{(x,y) \colon \norm{x+y} > 1/\sigma} \pr(sS\geq g(x)) \pr(sS\geq g(x+y)) \; \lambda(dx) \, \breve{\gamma}^{+}_{[2]}(dy) \\
  &\hspace*{10mm} \leq \pr(sS\geq h(1/\sigma)) \int_{\IR^2} \int_{\IR^2} \pr(sS\geq g(x+y)) \; \lambda(dx) \, \breve{\gamma}^{+}_{[2]}(dy) \\
  &\hspace*{25mm} {} + \pr(sS\geq h(1/\sigma)) \; \int_{\IR^2} \pr(sS\geq g(x)) \; \lambda(dx) \; \int_{\IR^2} \, \breve{\gamma}^{+}_{[2]}(dy) \\
  &\hspace*{10mm} \leq 2 \hbit \breve{\gamma}^{+}_{[2]}(\IR^2) \hbit M(s) \hbit \pr(sS\geq h(1/\sigma)) \\
  &\hspace*{10mm} \leq 2 \hbit \breve{\gamma}^{+}_{[2]}(\IR^2) \hbit M(t) \hbit K^{-1}t^{1/\beta}\sigma e^{-\sigma^2/(2\beta^2)} 
	 \stackrel{\sigma\to\infty}{\longrightarrow} 0,  
}
where the second last inequality uses again the translation invariance of $\lambda$.
\end{proof}

\begin{proof}[Proof of Theorem~\ref{thm:mainintro}]
The convergence results
Theorems~\ref{thm:fixprocconv} and~\ref{thm:randprocconv} combined with
 Lemma~\ref{lem:ppp} in Appendix~\ref{sec:proofppp}
and the mean convergence result Theorem~\ref{thm:meanlim} show that for each item, 
$N\vert_t$ converges in distribution
to a Poisson process with the appropriate mean measure. The process convergence follows from~\cite[Lemma~4.1]{Keeler2014}, read from \cite{Kallenberg1983}.
\end{proof}

\appendix

\section{Uniformly positive definite correlation functions}\label{sec:updfuncs}

We provide a number of tools for establishing uniform positive definiteness of correlation functions and give examples. For later reference, we consider correlation functions $\rho$ on general $d$-dimensional space here. Assume first that $\rho$ is continuous and stationary, i.e., $\rho(x,y)=\rho(x-y)$ depends only on the difference of its arguments $x,y \in \Rd$. By Bochner's Theorem there is then a unique symmetric probability measure $M$ on~$\Rd$ such that
\begin{equation*}
  \rho(z) = \int_{\Rd} e^{i x^{\top} z} \; M(dx) = \int_{\Rd} \cos(x^{\top} z) \; M(dx), \quad \text{$z \in \Rd$}.
\end{equation*}
Recall that the density $f$ (if it exists) of $M$ with respect to Lebesgue measure, is called the \emph{spectral density}. If $\rho$ is integrable, then $f$ can be obtained by the inverse Fourier transform
\begin{equation}  \label{eq:specdens}
  f(x) = \frac{1}{(2\pi)^{d}} \int_{\Rd} e^{-i x^{\top} z} \rho(z) \; dz = \frac{1}{(2\pi)^{d}} \int_{\Rd} \cos(x^{\top} z) \rho(z) \; dz, \quad \text{$x \in \Rd$}.
\end{equation}

The following theorem gives sufficient conditions for uniform positive definiteness. 
\begin{Theorem} \label{thm:upd}
Let $\rho \colon \Rd \to \R$ be a correlation function (i.e., symmetric and positive semidefinite) that is continuous and integrable. Assume that $\rho$ has a spectral density $f$ that is bounded away from $0$ on every ball $\bar B(0,H)$, $H > 0$.
Then $\rho$ is uniformly positive definite. The parameter $\delta$ in the definition can be chosen as
\ben{  \label{eq:mineigen}
  \delta = \delta(\eps) = \frac{\pi^{d/2}}{2^{d+1}\Gamma(\frac{d}{2}+1)} H^d f_0(H)
}
for arbitrary 
\be{
  H \geq \frac{24}{\eps} \biggl(\frac{\pi \Gamma^2(d/2+1)}{9}\biggr)^{1/(d+1)},
}
where $f_0(H) = \inf_{x \in \bar B(0,2H)} f(x)$.
\end{Theorem}
\begin{proof}
Let $\eps > 0$, $n \in \IN$ and $\tx_1, \ldots, \tx_n \in \R^d$ with $\min_{i \neq j} \|\tx_i-\tx_j\| \geq \eps$. For all $v \in \R^n$ we have by spectral decomposition 
\be{
  \sum_{i,j=1}^{n} v_i v_j \rho(\tx_i,\tx_j) \geq \lambda_{\tx} \|v\|^{2},
}
where $\lambda_{\tx}$ is the smallest eigenvalue of the matrix $(\rho(\tx_i,\tx_j))_{1 \leq i,j \leq n}$. By \cite[Theorem 12.3]{Wendland2005} the $\delta$ given in~\eqref{eq:mineigen} is a (uniform) lower bound on $\lambda_{\tx}$. To see this note that by~\eqref{eq:specdens} the Fourier transform $\widehat{\Phi}$ in \cite{Wendland2005} is just $(2\pi)^{d/2} f$.
Since $\delta$ is positive by the condition on $f$, the claim is shown.
\end{proof}

As examples we give two very general classes of isotropic (i.e. rotationally invariant) correlation functions that can be shown to be u.p.d.\ by the above theorem. We refer to \cite[Sections~2.7 and~5.2]{HandbookSpat2010} for more details on these and other examples.
\begin{Example}[The Mat\'{e}rn class of correlation functions and its boundaries]
This is arguably the most popular class in spatial statistics nowadays. We have
\begin{equation*}
  \rho(z) = \frac{1}{2^{\nu-1} \Gamma(\nu)} \biggl( \frac{\norm{z}}{\theta} \biggr)^{\nu} K_{\nu}(\norm{z}/\theta),
\end{equation*}
where $K_{\nu}$ is the modified Bessel function of the second kind, defined 
in \cite[(9.6.1)]{Abramowitz1964}.
Explicit formulae exist for $\nu \in \{k-\frac12 \colon k \in \NN\}$. For example, 
for $\nu = \frac{1}{2}$ we obtain the popular exponential correlation function, $\rho(z) = \exp(-\norm{z}/\theta)$.
The parameter $\nu > 0$ controls the regularity of the paths of a Gaussian random field with correlation function $\rho$, while $\theta>0$ controls the spatial scale on which there is substantial correlation. Various other parametrizations for the Mat\'{e}rn class exist. 
For general $\nu > 0$ the spectral density is given by
\begin{equation*}
  f(x) = \frac{\Gamma(\nu+d/2) \hbit \theta^d}{\Gamma(\nu) \hbit \pi^{d/2}} \hbit \bigl(1+\theta^2 \norm{x}^2\bigr)^{-(\nu+d/2)}.
\end{equation*}
So by Theorem~\ref{thm:upd} $\rho$ is uniformly positive definite and the factor $\delta$ can be computed explicitly. 
In particular note that $\delta(\eps)=\Omega(\eps^{2\nu})$ as $\eps\to0$.

Furthermore, according to \cite[9.7.2]{Abramowitz1964},
\be{
\lim_{r\to\infty} \frac{K_{\nu}(r)}{\sqrt{\frac{\pi}{2r}}e^{-r}}=1,
}
and so for our results above, we can take
\be{
\tilde\varrho(r)= \textrm{O}\left(r^{\nu-1/2} e^{-r/\theta}\right),
}
as $r\to\infty$. Thus Theorem~\ref{thm:mainintro} applies for these correlation functions.

Letting $\nu \to 0$, we obtain the so-called nugget correlation function $\rho(z) = 1\{z=0\}$, which is trivially u.p.d.\ with constant $\delta=1$.
In our setting, this corresponds to the case where the shadowing variables
are independent which is already understood.  From an applied point of view it is sometimes useful to replace $\theta$ by $\theta'/\sqrt{2\nu}$, which stabilizes $\rho$ as $\nu \to \infty$. Letting $\nu \to \infty$ in
this parameterization, we obtain the squared exponential (or Gaussian) correlation function $\rho(z) = \exp(-\norm{z}^2/(2{\theta'}^2))$ and we can set 
$\tilde\varrho(r)=\exp(-r^2/(2{\theta'}^2))$. The spectral density is given by
\begin{equation*}
  f(x) = \biggl( \frac{\theta'}{\sqrt{2\pi}} \biggr)^d \exp(-{\theta'}^2 \norm{x}^2/2)
\end{equation*}
and thus $\rho$ is u.p.d.\ again by Theorem~\ref{thm:upd} and $\delta$ 
can be taken to be  $\Theta(\eps^{-d} \exp\{-C_{\theta'}/\eps^2\})$ for some 
positive constant $C_{\theta'}$, which is not $\Omega(\eps^{c})$ as $\eps\to0$ for any $c>0$.
\end{Example}

\begin{Example}[Compactly supported polynomials of minimal degree]
  \cite{Wendland2005}, Section~9.3--9.4, studies isotropic positive semidefinite functions $\rho(z) = \rho_0(\norm{z})$, $z \in \Rd$, where $\rho_0 \colon \R \to \R$ is continuous, a polynomial on $[0,1]$, constant zero on $[1,\infty]$, and symmetrically extended to $\R$ (for the question of differentiability at zero). In particular the author constructs functions $\phi_{d,k}$ that satisfy the property that they have minimal degree among all $\rho_0$ of the above form that are in $C^{2k}(\R)$. See~\cite{Wendland2005}, Section~9.4, for concrete formulae.

Corollary~12.8 of \cite{Wendland2005} together with the considerations from the proof of Theorem~\ref{thm:upd} yield that for $k \in \NN$, $d \in \NN$ and for $k=0$, $d \geq 3$ the correlation function $\rho \colon \Rd \to \R$, given by
\begin{equation*}
  \rho(z) = \frac{\phi_{d,k}(\norm{z})}{\phi_{d,k}(0)},
\end{equation*}
is uniformly positive definite with $\delta = C_{d,k} \hbit \eps^{2k+1}$, where $C_{d.k}$ is a positive constant that can be computed in principle. Moreover, since the $\rho$ is only supported on a compact set, Theorem~\ref{thm:mainintro} applies for these correlation functions. Note also there is an analogous discussion and conclusion where the interval $[0,1]$ is replaced by  $[0,\theta]$ for  $\theta>0$.
\end{Example}

\section{Technical Poisson Convergence Results}\label{sec:proofppp}

\begin{proof}[Proof of Theorem~\ref{thmppp}]
Set $\mcx = [0,t]$ and recall that we denote by $\mfn$ the set of finite point configurations on $\mcx$ equipped with the usual $\sigma$-algebra $\mcn$. Let $\Upsilon \sim \pop(\lambda)$. 

Note that both metrics considered are of the form
\be{
  d \bigl( \law(\Psi \vert_t), \law(\Upsilon \vert_t)\bigr) = \sup_{f \in \mcf} \bigl| \E f(\Psi \vert_t) - \E f(\Upsilon \vert_t) \bigr|,
}
where $\mcf = \ftv = \{ 1_A \colon A \in \mcn\}$ is the set of measurable indicators if $d=\dtv$ and 
\be{
  \mcf = \fw = \bigl\{ \tilde{f} \colon \mfn \to \IR \hbit;\, | \tilde{f}(\psi)-\tilde{f}(\upsilon) | \leq \done(\psi,\upsilon) \text{ for all $\psi, \upsilon \in \mfn$} \bigr\}
}
is the set of $1$-Lipschitz-continuous functions with respect to the OSPA metric $\done$ if $d=\dtwo$.

We use Stein's method for Poisson process approximation as presented in \cite{Barbour1992a} (see e.g.\ the proof of Theorem 2.4 in that paper). For any function $f \in \mcf$ there is a function $h=h_f \colon \mfn \to \R$ so that we can equate $f(\cdot) - \E f(\Upsilont) = (\mathcal{A} h_f)(\cdot)$, where $\mathcal{A}$ is the generator of a spatial immigration-death process on $\mcx$ with immigration measure $\lambda\vert_t$ and unit per-capita death rate. 
Write
\ba{
  \| \Delta_1 h \| &= \sup_{\substack{x \in \mcx \\ \xi \in \mfn}} \, \bigl| h(\xi+\delta_x)-h(\xi) \bigr|, \\
  \| \Delta_2 h \| &= \sup_{\substack{x,y \in \mcx \\ \xi \in \mfn}} \, \bigl| h(\xi+\delta_x+\delta_y)-h(\xi+\delta_x)-h(\xi+\delta_y)+h(\xi) \bigr|.
}
By \cite[Lemma~2.2]{Barbour1992a} we obtain $\| \Delta_1 h \|, \| \Delta_2 h \| \leq 1$ if $f \in \ftv$, and by \cite[Proposition~4.1]{SX08} $\| \Delta_1 h \|, \| \Delta_2 h \| \leq \min\bigl(1,\frac{1+2\log^{+}(\lambda(t))}{\lambda(t)}\bigr) =: c(\lambda)$ if $f \in \fw$. The finiteness of these bounds justifies the integrability in the various statements below.

Noting that $\Psi \vert_t = \sum_{i \in \mcI} I_i \delta_{Y_i}$, we obtain
\ban{
  \E &f(\Psit) - \E f(\Upsilont) = \E \mathcal{A} h (\Psit) \nonumber\\[1mm]
  &= \E \int_0^t \bigl[ h(\Psit + \delta_s) - h(\Psit) \bigr]
    \; \lambda(ds) + \E \int_0^t  \bigl[ h(\Psit - \delta_s) - h(\Psit) \bigr] \; \Psit(ds)
    \nonumber\\[2mm]
  &= \E \sum_{i \in \mcI} \int_0^t \bigl[ h(\Psit + \delta_s) - h(\Psit) \bigr]
    \; \lambda_i(ds)
    + \E \sum_{i \in \mcI} I_i \bigl[ h(\Psit - \delta_{Y_i}) - h(\Psit) \bigr]
    \nonumber\\[2mm]
  &= \E \sum_{i \in \mcI} \int_0^t \Bigl( \bigl[ h(\Psit + \delta_s) - h(\Psit) \bigr]
    - \bigl[ h(\Psiti + \delta_s) - h(\Psiti) \bigr] \Bigr) \; \lambda_i(ds) \label{t1}\\
  &\hspace*{3mm} {} + \E \sum_{i \in \mcI} \int_0^t \bigl[ h(\Psiti + \delta_s) - h(\Psiti) \bigr] \; \lambda_i(ds) - \E \sum_{i \in \mcI} I_i \bigl[ h(\Psiti + \delta_{Y_i}) - h(\Psiti) \bigr] \label{t2}\\
  &\hspace*{3mm} {} + \E \sum_{i \in \mcI} \Bigl( I_i \bigl[ h(\Psiti + \delta_{Y_i}) - h(\Psiti) \bigr]
    - I_i \bigl[ h(\Psit) - h(\Psit - \delta_{Y_i}) \bigr] \Bigr), \label{t3}
}
where $\Psi^{(i)}\vert_t = \sum_{j \in A_i^c} I_i \delta_{Y_i}$.
We bound the absolute values of the three terms on the right hand side. 
For the absolute value of~\eqref{t1}, we obtain 
\ba{
  \biggl| \E \sum_{i \in \mcI} &\int_0^t \Bigl( \bigl[ h(\Psit + \delta_s) - h(\Psit) \bigr]
    - \bigl[ h(\Psiti + \delta_s) - h(\Psiti) \bigr] \Bigr) \; \lambda_i(ds) \biggr| \\
  &\leq \E \biggl( \sum_{i \in \mcI} \int_0^t \| \Delta_2 h \| \biggl( \sum_{j \in A_i} I_j \biggr) \; \lambda_i(ds) \biggr)
  = \| \Delta_2 h \| \sum_{i \in \mcI, j \in A_i} p_i p_j,
}
where the inequality follows by adding the points of $\Psit-\Psiti$ individually.
In much the same way, the absolute value of~\eqref{t3} is bounded as
\ba{
  \biggl| \E \sum_{i \in \mcI} &I_i \Bigl( \bigl[ h(\Psiti + \delta_{Y_i}) - h(\Psiti) \bigr]
    - I_i \bigl[ h(\Psit) - h(\Psit - \delta_{Y_i}) \bigr] \Bigr) \biggr| \\
  &\leq \E \biggl( \sum_{i \in \mcI} I_i \, \| \Delta_2 h \| \, \sum_{j \in A_i \setminus \{i\}} I_j \biggr)
  = \| \Delta_2 h \| \sum_{\substack{i \in \mcI, j \in A_i\\ i \neq j}} p_{ij}.
}

In order to estimate the absolute value of~\eqref{t2}, write $\lambda_i(\cdot \mvert \mcf_i)$ for the regular conditional distribution of $Y_i$ given $\mcf_i$. By disintegration (e.g.\ Theorem~6.4 in \cite{Kallenberg2002}), since $\Psiti$ is $\mcf_i$-measurable,
\ba{
  \E \Bigl( &I_i \bigl[ h(\Psiti + \delta_{Y_i}) - h(\Psiti) \bigr] \Bigr) \\
  &= \E \Bigl( \E \Bigl( \I[Y_i \leq t] \bigl[ h(\Psiti + \delta_{Y_i}) - h(\Psiti) \bigr] \Bigm| \mcf_i \Bigr) \Bigr) \\
  &= \E \int_0^\infty \I[s \leq t] \bigl[ h(\Psiti + \delta_s) - h(\Psiti) \bigr] \; \lambda_i(ds \mvert \mcf_i) \\
  &= \E \int_0^t \bigl[ h(\Psiti + \delta_s) - h(\Psiti) \bigr] \; \lambda_i(ds \mvert \mcf_i).
}

For any $\xi \in \mfn$ define $g_{\xi} \colon \mcx \to \R$ by $g_{\xi}(s) = h(\xi+\delta_s)-h(\xi)$, $s \in \mcx$. Note that $|g_{\xi}(s)| \leq \|\Delta_1 h\|$ for all $s \in \mcx$.
If $f \in \fw$ (and in fact for a somewhat larger class of functions) it was shown in~\cite{Schuhmacher2005a} in the argument that leads up to Inequality~(2.3) that $g_{\xi}$ is Lipschitz continuous with constant $\bigl(1,\frac{1+\log^{+}(\lambda(t))}{\lambda(t)}\bigr) \leq c(\lambda)$.

In summary we obtain for the absolute value of~\eqref{t2} and general $f \in \mcf$,
\ban{
  \biggl|\E \sum_{i \in \mcI} &\int_0^t \bigl[ h(\Psiti + \delta_s) - h(\Psiti) \bigr] \; \lambda_i(ds) - \E \sum_{i \in \mcI} I_i \bigl[ h(\Psiti + \delta_{Y_i}) - h(\Psiti) \bigr] \biggr| \notag\\
  &= \biggl| \E \sum_{i \in \mcI} \biggl( \int_0^t g_{\Psiti}(s) \; \lambda_i(ds) - \int_0^t g_{\Psiti}(s) \; \lambda_i(ds \mvert \mcf_i) \biggr) \biggr| \label{eq:t2cont}
}

If $\mcf=\ftv$, we simply have
\ba{
  \biggl| \int_0^t &g_{\Psiti}(s) \; \lambda_i(ds) - \int_0^t g_{\Psiti}(s) \; \lambda_i(ds \mvert \mcf_i) \biggr| \\
  &\leq \int_0^t \bigl| g_{\Psiti}(s) \bigr| \; \Bigl| \lambda_i(\cdot) - \lambda_i(\cdot \mvert \mcf_i) \Bigr| (ds)\\[1mm]
  &\leq \bigl\| \lambda_i\vert_t(\cdot) - \lambda_i\vert_t(\cdot \mvert \mcf_i) \bigr\|_{\mathrm{TV}},
}
where the long absolute value bars in the second line denote the variation of the signed measure.

If $\mcf=\fw$, we base the upper bound on a somewhat more general result. Let $g \colon [0,t] \to \R$ be a Lipschitz continuous function with constant $c_2$ and $\| g \|_{\infty} \leq c_1$. For arbitrary probability measures $\lambda_1$, $\lambda_2$ on $\R_{+}$ with c.d.f.s $F_1$, $F_2$, we let $(X,Y)$ be a quantile coupling of $\lambda_1$ and $\lambda_2$, i.e., let $U \sim U[0,1]$ and set $X = F_1^{-1}(U)$, $Y = F_2^{-1}(U)$, where $F_k^{-1}(u) = \inf \{x \in \R_{+} \colon F_k(x) \geq u\}$ is the generalized inverse.   
Assuming $F_1(t) \leq F_2(t)$ (otherwise switch $F_1$ and $F_2$), we have $\pr(X \leq t, Y > t) = 0$ and thus
\ban{
  \biggl| \int_0^t &g(s) \; \lambda_1(ds) - \int_0^t g(s) \; \lambda_2(ds) \biggr| \notag\\
  &= \bigl| \E \bigl( g(X) \I[X \leq t] - g(Y) \I[Y \leq t] \bigr) \bigr| \notag\\[1mm]
  &= \bigl| \E \bigl( (g(X) - g(Y)) \I[X \leq t,\, Y \leq t] \bigr)  - \E \bigl( g(Y) \I[X > t,\, Y \leq t] \bigr) \bigr| \notag\\[1mm]
  &\leq c_2 \hbit \E \bigl( |X-Y| \, \I[U \leq F_1(t)] \bigr) + c_1 \hbit \pr(F_1(t) < U \leq F_2(t)) \notag\\
  &\leq c_2 \int_0^t |F_1(s)-F_2(s)| \; ds + c_1 |F_1(t) - F_2(t)|, \label{eq:t2help}
}
where the last inequality holds because we see by comparing the areas between the quantile functions and the distribution functions and using $F_1(t) \leq F_2(t)$ that
\be{
  \int_0^{F_1(t)} \bigl| F_1^{-1}(u) - F_2^{-1}(u) \bigr| \; du \leq \int_0^{t} \bigl| F_1(s) - F_2(s) \bigr|  \, \; ds.
}
Applying Inequality~\eqref{eq:t2help} to the right hand side of \eqref{eq:t2cont} yields the last two summands of the upper bound claimed.
\end{proof}

\begin{Lemma}\label{lem:ppp}
If $\Upsilon_1, \Upsilon_2$ are two finite Poisson point processes on $[0,t]$ with intensity measures
defined by non-decreasing, right-continuous functions $\Lambda_1, \Lambda_2$,
such that $\Lambda_1(0)=\Lambda_2(0)=0$, then
\be{
\dtwo(\law(\Upsilon_1), \law(\Upsilon_2))\leq \int_0^t \bigl|\Lambda_1(s)-\Lambda_2(s)\bigr| \; ds
	+\bigl|\Lambda_1(t)-\Lambda_2(t)\bigr|.
}
\end{Lemma}
\begin{proof}
Following the first part of the proof of Theorem~\ref{thmppp},
we have that for an appropriate family $\mathcal{H}$ of functions $h \colon \mfn \to \IR$,
\ba{
\dtwo(\law(&\Upsilon_1), \law(\Upsilon_2))\\
&= \sup_{h\in\mathcal{H}}\left|\mean\int_0^t[h(\Upsilon_2+\delta_s)-h(\Upsilon_2)] \; \Lambda_1(ds) 
+ \mean\int_0^t[h(\Upsilon_2-\delta_s)-h(\Upsilon_2)] \; \Upsilon_2(ds)\right|\\
 &=\sup_{h\in\mathcal{H}}\left|\mean\int_0^t[h(\Upsilon_2+\delta_s)-h(\Upsilon_2) ] \; \Lambda_1(ds) - \mean\int_0^t[h(\Upsilon_2+\delta_s)-h(\Upsilon_2)] \; \Lambda_2(ds)\right|,
}
where we have used that $\Upsilon_2$ has the stationary distribution
of a spatial immigration-death process with immigration measure $\Lambda_2$ and unit per-capita death rate. As in the proof of Theorem~\ref{thmppp}, for each $h\in\mathcal{H}$, we can write 
\be{
h(\Upsilon_2+\delta_s)-h(\Upsilon_2)=g_{\Upsilon_2}(s)=:g(s),
}
where the $g_{\Upsilon_2}$ is bounded by $1$ with Lipschitz constant $1$, uniform over $\Upsilon_2$.
Thus for $h\in\mathcal{H}$, denoting 
 generalized inverses 
by $\Lambda_i^{-1}$, assuming without loss that $\Lambda_2(t)\geq \Lambda_1(t)$,
and using standard results about the Lebesgue--Stieltjes integral,
we have
\ba{
&\left|\mean\int_0^t[h(\Upsilon_2+\delta_s)-h(\Upsilon_2) ] \; \Lambda_1(ds) - \mean\int_0^t[h(\Upsilon_2+\delta_s)-h(\Upsilon_2)] \; \Lambda_2(ds)\right| \\
&\qquad=
\left|\mean\left[\int_0^t g(s) \; \Lambda_1(ds) - \int_0^t g(s) \; \Lambda_2(ds)\right]\right|\\
&\qquad=\left|\mean\left[\int_0^{\Lambda_1(t)} g(\Lambda_1^{-1}(u)) \; du - \int_0^{\Lambda_2(t)} g(\Lambda_2^{-1}(u)) \; du\right]\right| \\
&\qquad \leq \left|\mean\left[\int_0^{\Lambda_1(t)} [g(\Lambda_1^{-1}(u))-g(\Lambda_2^{-1}(u))] \; du
\right]\right|+\left|\mean\left[\int_{\Lambda_1(t)}^{\Lambda_2(t)} g(\Lambda_2^{-1}(u)) \; du\right]\right| \\
&\qquad \leq \int_0^{\Lambda_1(t)} \bigl|\Lambda_1^{-1}(u)-\Lambda_2^{-1}(u)\bigr| \; du
+\left|\Lambda_2(t)-\Lambda_1(t)\right|,
}
which is bounded by what is claimed.
\end{proof}

\section*{Acknowledgments}

Yu-Hsiu Paco Tseng assisted in running simulations, funded by a University of Melbourne  Vacation Scholarship.
Simulations were run in {\tt R} \cite{R}
using the following packages: 
{\tt fields} \cite{fields}, {\tt hexbin} \cite{hexbin}, {\tt pracma} \cite{pracma}, {\tt RandomFields}  \cite{Schlather2015}, {\tt spatstat} \cite{Baddeley2005},
{\tt stpp} \cite{stpp}.  Figures~\ref{fig:CDFPPP} and~\ref{fig:CDFHex} were produced using 
the {\tt ggplot2} package \cite{Wickham2009}.
NR received support from ARC grant DP150101459 and thanks
the Institute for Mathematical Stochastics at
Georg-August-University G\"ottingen for supporting a visit 
in 2015 during which work on this article took place. This work commenced while the authors were visiting the Institute for Mathematical Sciences, National University of Singapore in 2015, supported by the Institute. We thank the referee and AE for helpful comments which have greatly improved the practical aspects of the paper.

\end{document}